\documentclass[a4paper,notitlepage,fleqn,12pt]{article}
\usepackage[tbtags]{amsmath}
\usepackage{amsfonts}
\usepackage{amsthm}
\usepackage{color}
\usepackage{graphicx,amssymb}
\usepackage{multirow}
\usepackage{rotating}
\usepackage{mathrsfs}
\usepackage{epstopdf}
\usepackage{enumerate}
\usepackage{pdflscape}
\usepackage{geometry}
\usepackage{natbib}
\usepackage{changepage}
\usepackage{dcolumn}
\usepackage{tabularx}
\usepackage{subfig}
\usepackage{mathabx}
\usepackage[mathscr]{euscript}

\usepackage{floatrow}
\usepackage{tikz}
\usetikzlibrary{graphs}
\usepackage{mathdots}
\usepackage{algorithm}
\usepackage{algorithmicx}
\usepackage{algpseudocode}
\newcolumntype{d}[1]{D{.}{.}{#1}}
\newcolumntype{Y}{>{\raggedleft\arraybackslash}X}
\newcolumntype{Z}{>{\centering\arraybackslash}X}
\newtheorem{lemma}{Lemma}

\newtheorem{assum}{Assumption}
\newtheorem{thm}{Theorem}
\newtheorem{remark}{Remark}

\DeclareMathOperator*{\diag}{diag}

\newcommand{\bm}[1]{\mbox{\boldmath{$#1$}}}
\numberwithin{equation}{section}

\begin{document}

\title{\vspace{-10mm} High-Frequency-Based Volatility  Model with Network Structure}
\author{Huiling Yuan$^a$,
	Guodong Li$^a$,
	Junhui Wang$^b$\footnote {Corresponding author: Junhui Wang. Address: 83 Tat Chee Avenue, Kowloon Tong, Hong Kong. Tel: 852-3442 2153. E-mail: j.h.wang@cityu.edu.hk.} \\
	\\$^a$ \small{Department of Statistics \& Actuarial Science,
		The University of Hong Kong, Hong Kong, China}\\
	$^b$ \small { School of Data Science, City University of Hong Kong, Hong Kong, China}\\
}

\date{}
\maketitle

\begin{abstract}
	
This paper introduces one new multivariate volatility model that can accommodate an appropriately defined network structure based on low-frequency and high-frequency data.
The model reduces the number of unknown parameters and the computational complexity substantially. The model parameterization and iterative multistep-ahead forecasts are discussed and the targeting reparameterization is also presented.  Quasi-likelihood functions for parameter estimation are proposed and their asymptotic properties are established. A series of simulation experiments are carried out to assess the performance of the estimation in finite samples. An empirical example is demonstrated that the proposed model outperforms the network GARCH model, with the gains being particularly significant at short forecast horizons.

\end{abstract}


\textit{MOS subject classification:} 62M10, 62M20, 62F12
	
\newpage

\section{Introduction}

Volatility analysis is an important issue in modern financial markets. As such, a great many researchers have focused on developing and evaluating volatility models. The natural information source of the volatility is the historical data of the security, which can be further divided into low-frequency and high-frequency historical data. The low-frequency historical data are referred to as observed prices of the security daily or at longer time horizons.  The autoregressive conditional heteroskedasticity (ARCH)  in \cite{Engle:1982} and the generalized autoregressive conditional heteroskedasticity (GARCH) models in \cite{Bollers:1986} are the most famous models for the analysis of low-frequency data.  The high-frequency data are referred to as the intra-day prices of the security, such as tick-by-tick, 1-second, 5-minute data and etc. The researchers often model the high-frequency historical data by continuous-time It\^{o} processes and develop realized volatility estimators. These estimators include two-time scale realized volatility  \citep{Zh:2005}, multi-scale realized volatility (MSRV) \citep{Zh:2006}, kernel realized volatility  \citep{B:2008}, pre-averaging realized volatility \citep{J:2009} and quasi-maximum likelihood estimator  \citep{X:2010}.

However, these models were developed for high-frequency and low-frequency data quite separately. In fact, high-frequency and low-frequency data must be inter-related at different time scales, due to just these different time scales. There are some attempts to bridge the gap between high-frequency and low-frequency data. \cite{Wang:2002} studied the statistical relationship between the GARCH and diffusion model; \cite{Engle:2006} proposed the multiplicative error model;  \cite{Ghysels:2006} studied the Mixed Data Sampling model (MIDAS);  \cite{Corsi:2009} provided Heterogeneous Autoregressive model for Realized Volatility (HAR-RV) model;  \cite{Shephard:2010} proposed the high-frequency-based volatility models,  they call `HEAVY models' (High-frEquency-bAsed VolatilitY models) for simplicity; \cite{Hansen:2012} studied volatilities analysis by combining the realized GARCH model and the high-frequency volatility model;  \cite{KW:2016} proposed GARCH-It\^o model for merged low-frequency data and high-frequency data.  \cite{song2020realized} and \cite{Yuancui:2022} extended GARCH-It\^o model to Realized GARCH-It\^o model  and GQARCH-It\^o model, respectively.

All the findings show that the volatility models combining low-frequency and high-frequency data have stronger forecasting power than those based on one information source. 
Among them, \cite{Shephard:2010} proposed the high-frequency-based volatility (HEAVY) model, which is designed to harness high-frequency data to make multistep-ahead predictions of the volatility of returns. Empirical studies show that the HEAVY model outperforms the GARCH model in- and out-of-sample performance. However, portfolio optimization and risk management in the financial markets often require cross-sections of hundreds of different stocks. This led to the development of multivariate high-frequency-based volatility (multivariate HEAVY) model \citep{Noureldin:2012}. They discussed the differences between multivariate HEAVY  and multivariate GARCH models, and empirical results suggest that the multivariate HEAVY model outperforms the multivariate GARCH model. Moreover, \cite{JinandMaheu:2013}  provided a new approach to improve the density forecasts of multivariate daily returns, however, the multivariate HEAVY model is much easier to estimate and allows for straightforward out-of-sample model evaluation since the multivariate HEAVY model provides closed-form forecasting formulas.

Due to the structure of the multivariate HEAVY model, it also suffers from the ``curse of dimensionality", not only in terms of the number of parameters in the model but also in the dimension of the realized measure required to drive the dynamics.  \cite{sheppard:2019} proposed the factor HEAVY models. Dimension reduction for estimating large covariance matrices is achieved by imposing a factor structure with time-varying conditional factor loadings. They also compared the factor HEAVY and DCC GARCH models by empirical studies. In this work, we develop a new method for a multivariate HEAVY model, based on network structure data. We call the proposed model the network HEAVY model.

As we all know, in the real world, many systems can be described by complex networks \citep{RekaandBarabasi:2002, Newman:2003}. In essence, a financial market can be represented as a network, where nodes represent financial entities such as stocks, and the edges connecting them represent the correlations between their returns \citep{Mantegna:1999, Bonanno:2001, Huang:2009, Tumminello:2010}.  Empirical data analysis also shows that model performance can be improved significantly by incorporating network structure information \citep{Goel:2014, Nitzan:2011, Wei:2016}.
Moreover, in order to incorporate the network information among individuals, \cite{Zhu:2017} developed a network vector autoregression (NAR) model. The response of each individual can be explained by its lagged values, the average of its neighbors, and a set of node-specific covariates. However, all the individuals are assumed to be homogeneous since they share the same autoregression coefficients. To express individual heterogeneity, \cite{Zhu:2020} developed a grouped NAR (GNAR) model.  Individuals in a network can be classified into different groups, characterized by different sets of parameters.  In order to overcome the high-dimensionality, \cite{Zhou:2020} proposed a network GARCH  model that used information derived from an appropriately defined network structure. This reduces the number of parameters and the computational complexity of the network GARCH model.

Inspired by the existing literature, the network HEAVY model is proposed.  A financial market index is considered. For example,  CSI300 Index is a stock market index that measures the stock performance of 300 large companies listed on stock exchanges in the Shanghai and Shenzhen stock markets in China. It is one of the most commonly followed equity indices, and one of the best representatives of the China stock markets. To describe the network structure, we define an adjacency matrix $A=(a_{i,j}), i, j=1, \cdots, N$, where $a_{i,j}=1$, if the  $i$-th stock  and the $j$-th stock belong to the same economic sector, and $a_{i,j}=0$ otherwise.

The proposed model is different from the existing ones in the following ways. First, in contrast to the traditional multivariate HEAVY model, the network HEAVY model uses information on the stocks by defining an appropriate network structure. Secondly, the computational complexity drops from  $\mathcal{O}(N^2)$ to $\mathcal{O}(N)$. Thirdly, a quasi-maximum likelihood estimator (QMLE) is given to estimate the network HEAVY model, and its asymptotic properties are also established. Fourthly, compared to the network GARCH model, the network HEAVY model has stronger volatility forecasting power for one-day volatility forecasting, two-day volatility forecasting, or forecasting at longer periods.

This paper is organized as follows. Section 2 introduces the network HEAVY model, as well as the model parameterization, iterative multistep-ahead forecasts, and targeting reparameterization. Section 3 derives the asymptotic properties. Section 4 conducts a series of simulation experiments to check the finite sample performance of the proposed methodology.
Section 5 carries out an empirical analysis with 18 constituent stocks from the CSI 300 to demonstrate the usefulness of the proposed model in short-term volatility forecasting.
Section 6 concludes the paper. All the proofs are relegated to Appendix.

\section{Network HEAVY model }

The standard HEAVY models are designed to harness high-frequency data to make multistep-ahead predictions of the volatility of returns. These models allow for both mean reversion and momentum,
\begin{align}
	\mathrm{var}(r_{t}|\mathcal{F}_{t-1}^{\rm HF})=h_t=&\omega+\alpha {\rm RM}_{t-1}+\beta h_{t-1}, \ \omega, \alpha \geq 0, \  \beta \in [0, 1), \label{HEAVY-r}\\
	\mathrm{E}({\rm RM}_{t}|\mathcal{F}_{t-1}^{\rm HF})=\mu_{t}=&\omega_{R}+\alpha_{R} {\rm RM}_{t-1}+\beta_{R} \mu_{t-1}, \ \omega_R, \alpha_R, \beta_R \geq 0, \ \alpha_R + \beta_R \in [0, 1), \label{HEAVY-rm}
\end{align}
where $\{r_{t}\}$ is the daily financial return sequence and $\{{\rm RM}_t\}$ is a corresponding sequence of daily realized measures. $\mathcal{F}_{t-1}^{\rm HF}$ denotes the past information of $r_{t}$ and ${\rm RM}_{t}$, and the sup-script $\rm HF$ denotes the high-frequency dataset. Particularly, (\ref{HEAVY-r}) is called the HEAVY-r model and (\ref{HEAVY-rm}) the HEAVY-RM model. It is noted that the  HEAVY model is non-nested.

Supposed that there are $N$ stocks, and let $r_{it}$ denote the return of $i$-th stock at time $t$, where $i=1,\ldots,N$ and $t=1,\ldots,T$. We further define an adjacency matrix $A=(a_{ij}) \in \mathbb{R}^{N\times N}$, where  $a_{ij}=a_{ji}= 1$ if  $i$-th and $j$-th stocks are related, and $a_{ij}=0$ otherwise. For example, two stocks are connected if they belong to the same economic sector. For convenience, we set $a_{ii}=0$  for $1 \leq i \leq N$. To model this network structure among stocks, we proposed the network HEAVY (NHEAVY) model,
\begin{align}
	\mathrm{var}(r_{it}|\mathcal{F}_{t-1}^{\rm HF}) = h_{it}=&\omega+\alpha {\rm RM}_{it-1}+\lambda d_{i}^{-1}\sum_{j\neq i}a_{ij}{\rm RM}_{jt-1}+\beta h_{it-1}, 	\label{NHEAVY-r}\\
	\mathrm{E}({\rm RM}_{it}|\mathcal{F}_{t-1}^{\rm HF}) = \mu_{it}=&\omega_{R}+\alpha_{R} {\rm RM}_{it-1}+\lambda_{R} d_{i}^{-1}\sum_{j\neq i}a_{ij}{\rm RM}_{jt-1}+\beta_{R} \mu_{it-1}, 
	\label{NHEAVY-rm}
\end{align}
where $\sum_{j\neq i}$ represents $\sum_{j=1, j\neq i}^{N}$, and
$d_{i}=\sum_{j=1}^{N}a_{ij}$ is the total number of neighbors of the $i$-th stock, which is the out-degree. In contrast to the HEAVY model, the $i$-th stock is assumed to be affected by its directly connected neighbors, which is particularly reasonable in practice. Similarly, we call (\ref{NHEAVY-r}) NHEAVY-r model and (\ref{NHEAVY-rm}) NHEAVY-RM model. Here, $\lambda$ and $\lambda_{R}$ are used to capture the average influence of other stocks on the $i$-th stock over the network structure. Note that the NHEAVY model only requires $\mathcal{O}(N)$ parameters to model the dependency structure among $N$ stocks, similar discussions have been introduced in the network vector autoregression (NAR) model \citep{Zhu:2017} and the network GARCH model \citep{Zhou:2020}. 

Moreover, we discuss the stationary condition for the proposed NHEAVY model. 
Define $\bm{r}_t^2=({r}_{1t}^2, \ldots, {r}_{Nt}^2)^{T}$, $\mathbf{RM}_t=({\rm RM}_{1t}, \ldots, {\rm RM}_{Nt})^{T}$, $\mathbf{h}_t=(h_{1t}, \ldots, h_{Nt})^{T}$, $\bm {\mu}_t=(\mu_{1t}, \ldots, \mu_{Nt})^{T}$, $\mathbf{ D}=\diag(d_{1}, \ldots, d_{N})$, $\bm {\omega}=(\omega, \ldots, \omega)^{T}$, $\bm {\omega}_{R}=(\omega_{R}, \ldots, \omega_{R})^{T}$, $\bm {\beta}=\diag (\beta, \ldots, \beta)$, $\bm {\beta}_{R}=\rm diag (\beta_{R}, \ldots, \beta_{R})$, $\bm {\alpha}=\diag(\alpha, \ldots, \alpha)$, $\bm {\alpha}_{R}=\diag(\alpha_{R}, \ldots, \alpha_{R})$,
$\bm{ \lambda}=\diag(\lambda, \ldots, \lambda)$, $\bm {\lambda}_{R}=\rm diag(\lambda_{R} \ldots, \lambda_{R})$, and $\mathbf {A}=(a_{i_{1} i_{2}}) \in \mathbb{R}^{N \times N}$. 
Following the NHEAVY representation in (\ref{NHEAVY-r}) and (\ref{NHEAVY-rm}), the dynamic structure of the bivariate model can be gleaned from rewriting
\begin{align}
	\begin{pmatrix}
		\mathbf {h}_t\\
		\bm {\mu}_t
	\end{pmatrix}
	&=\mathbf {w}+\mathbf {B}\begin{pmatrix}
		\mathbf {h}_{t-1}\\
		\bm {\mu}_{t-1}
	\end{pmatrix}+\begin{pmatrix}
		\bm{ \alpha}+\bm {\lambda} \mathbf {D}^{-1}\mathbf {A}&\quad \mathbf{0}\\
		\mathbf{0}&\quad \bm {\alpha}_{R}+\bm{ \lambda}_{R} \mathbf {D}^{-1}\mathbf {A}
	\end{pmatrix}(\mathbf{RM}_{t-1}-\bm {\mu}_{t-1}),
\label{vector-form}
\end{align}
where $\mathbf {w}= \begin{pmatrix}
	\bm{ \omega}\\
	\bm {\omega}_{R}
\end{pmatrix}$, 
$\mathbf {B}= \begin{pmatrix}
	\bm{ \beta }& \quad \bm {\alpha} +\bm {\lambda }\mathbf {D}^{-1}\mathbf {A}\\
	\mathbf{0} & \quad \bm {\alpha}_{R}+\bm {\lambda}_{R} \mathbf {D}^{-1}\mathbf {A}+\bm {\beta}_R
\end{pmatrix}$.

 Let $\bm {\varepsilon}_{t}= ({\varepsilon}_{1t}, \ldots, {\varepsilon}_{Nt})^{\top}$, ${\varepsilon}_{it}$ is i.i.d with $\mathrm{E} \{{ \varepsilon}_{it}|\mathcal {F}_{t-1}^{\rm HF}\}=1$  and $\mathrm{E} \{{ \varepsilon}_{it}^2|\mathcal {F}_{t-1}^{\rm HF}\}={{\kappa}}_2^{r}$. Similarly, ${\bm \epsilon}_{t}=(\epsilon_{1t}, \ldots, \epsilon_{Nt})^{\top}$, ${ \epsilon}_{it}$ is i.i.d with $\mathrm{E} \{{ \epsilon}_{it}|\mathcal {F}_{t-1}^{\rm HF}\}= {1}$  and $\mathrm{E} \{{ \epsilon}_{it}^2|\mathcal {F}_{t-1}^{\rm HF}\}={ {\kappa}}_2^{R}$.  Furthermore, $\mathrm{E} \{{ \varepsilon}_{it}{ \epsilon}_{it}|\mathcal {F}_{t-1}^{\rm HF}\}={{\kappa}}_2^{r,R}$, $\mathrm{E} \{{ \varepsilon}_{it}{ \epsilon}_{jt}|\mathcal {F}_{t-1}^{\rm HF}\}=0, i \neq j$.
Denote ${r}_{it}^2=\varepsilon_{it}{h}_{it}$ \citep{Engle:2006}, ${\rm RM}_{it}^2=\epsilon_{it}{\mu}_{it}$ \citep{ Shephard:2010}. Let $\mathbf{u}_t=({u}_{1t}, \ldots, {u}_{Nt})^{\top}$, ${u}_{it}={r}_{it}^2- {h}_{it}=({\varepsilon}_{it}-1){h}_{it}$, and ${\mathbf{u}_{R}}_{t}= ({{u}_{1Rt}}, \ldots, {{u}_{NRt}})^{\top}$, ${{u}_{iRt}}={\rm RM}_{it}-{\mu}_{it}=(\epsilon_{it}-1){\mu}_{it}$. Then, both $\{{u}_{it}\}$ and $\{{u}_{iRt}\}$ can be viewed as i.i.d ``noise" with mean zero, and finite variance.

From (\ref{vector-form}), a biproduct of the process ${\bm{r}}_t^2$ and $\mathbf{RM}_t$ can be the VARMA(1,1) representation,
\begin{align}
		\begin{pmatrix}
			\bm{r}_t^2\\
			\mathbf{RM}_t
		\end{pmatrix}=&\mathbf {w}+\mathbf {B}
		\begin{pmatrix}
			\bm{r}_{t-1}^2\\
     	\mathbf{RM}_{t-1}
	\end{pmatrix}
+\begin{pmatrix}
	(\mathbf{1}_N-\beta \mathbf{1}_N L)\mathbf{u}_t\\
(\mathbf{1}_N-\beta_{R} \mathbf{1}_N L){\mathbf{u}_{R}}_{t}
\end{pmatrix},\label{VARMA}
\end{align}
where $\mathbf{1}_N=(1,\ldots, 1)^{\top}$, and $L$ is the lag operator.

By the VARMA(1,1) representation in (\ref{VARMA}), we immediately have the necessary and sufficient condition for the existence of a unique strictly stationary solution to the NHEAVY models. Strictly speaking, the eigenvalues of $\mathbf{B}$ must be less than one in modulus. Since $\mathbf{B}$ is block triangular, its eigenvalues are members of the multiset of the eigenvalues of $\bm{ \beta }$ and $\bm {\alpha}_{R}+\bm {\lambda}_{R} \mathbf {D}^{-1}\mathbf {A}+\bm {\beta}_R$. Suppose that $\lambda_d$ is any arbitrary eigenvalue of $\mathbf{D}^{-1} \mathbf {A}$. Since $\mathbf{D}^{-1} \mathbf {A}=\left(a_{i j} / d_{i}\right)$, by the Gershgorin circle theorem on the eigenvalues of matrices and by the definition of $d_{i}$, we can get that $|\lambda_d| \leq 1$. Thus, the largest eigenvalue in modular of $\bm {\alpha}_{R}+\bm {\lambda}_{R} \mathbf {D}^{-1}\mathbf {A}+\bm {\beta}_R$ is smaller than $\alpha_{R}+\lambda_{R}+\beta_{R}$, i.e., $\rho\left(\bm {\alpha}_{R}+\bm {\lambda}_{R} \mathbf {D}^{-1}\mathbf {A}+\bm {\beta}_R\right) \leq \alpha_{R}+\lambda_{R}+\beta_{R}<1$. The largest eigenvalue of $\bm{ \beta }$,  i.e., $\rho(\bm \beta)=\beta<1$. 

Multistep-ahead forecasts of volatility are very important for asset allocation and risk assessment, since these tasks are usually carried out over multiple days. In contrast to one-step-ahead forecasts of volatility that only require the NHEAVY-r model, both NHEAVY-r and NHEAVY-RM models play a central role in multistep-ahead forecasts.

\begin{remark}
Let $\mathbf {I}$ denote an $N\times N$ identity matrix,  for $s\geq 0$, from the dynamic structure representation (\ref{vector-form}), we have the multistep-ahead forecasts as follows,
\begin{align}
	\begin{pmatrix}
		\mathbf {h}_{t+s|t-1}\\
		\bm {\mu}_{t+s|t-1}
	\end{pmatrix}=(\mathbf {I}+\mathbf {B}+\cdots+\mathbf {B}^{s})\mathbf {w}+\mathbf {B}^{s+1}\begin{pmatrix}
		\mathbf {h}_{t-1}\\
		\bm {\mu}_{t-1}
	\end{pmatrix},
	\label{multi-step forecasting}
\end{align}
Write
$\mathbf {v}_{1}=\bm{ \alpha}+\bm {\lambda} \mathbf {D}^{-1}\mathbf {A}$ and $\mathbf {v}_{2}=\bm {\alpha}_{R}+\bm{ \lambda}_{R} \mathbf {D}^{-1}\mathbf {A}+\bm {\beta}_R$, then for $J=1, 2, \ldots, $
\begin{align*}
	\mathbf {B}^{J}=\begin{pmatrix}
		\bm {\beta}^{J} &\quad \mathbf {v}_{1} (\mathbf {v}_{2}^{J-1}+\mathbf {v}_{2}^{J-2}\bm {\beta}+\dots+\bm {\beta}^{J-1})\\
		\mathbf{0}&\quad \mathbf {v}_{2}^{J}
	\end{pmatrix}.
\end{align*}  
\label{remark1}
\end{remark}
Similar to the HEAVY model, it is advantageous to re-parameterize the NHEAVY model, therefore, the intercepts are explicitly related to the unconditional mean of squared returns and realized measures. Let $\bm \mu=(\mu_1, \ldots, \mu_N)^{\top}$ and $\bm {\mu}_{R}=({{\mu}_{R}}_{1}, \ldots,  {{\mu}_{R}}_{N} )^{\top}$, where $\mu_i=\mathrm{E}({r}_{it}^2)$ and ${{\mu}_{R}}_{i}=\mathrm{E}({\rm{RM}}_{it})$.
\begin{remark}
The ``targeting parameterization" for the NHEAVY model can be provided as follows,
\begin{align}
			&\mathbf {h}_t
		=(\mathbf{ I}-(\bm \alpha+\bm \lambda \mathbf {D}^{-1}\mathbf {A })\bm \kappa-\bm \beta)\bm \mu+(\bm \alpha+\bm \lambda \mathbf {D}^{-1}\mathbf {A} )  \mathbf{RM}_{t-1}+\bm \beta \mathbf {h}_{t-1}, \label{target-1}\\
		&\bm {\mu}_{t}
		=(\mathbf {I}-\bm {\alpha}_{R} -\bm {\lambda}_{R} \mathbf {D}^{-1}\mathbf {A}-\bm {\beta}_{R})\bm {\mu}_{R}+(\bm {\alpha}_{R}+\bm {\lambda}_{R} \mathbf {D}^{-1}\mathbf {A}) \mathbf{RM}_{t-1}+\bm {\beta}_{R} \bm {\mu}_{t-1}, \label{target-2}
\end{align}
where $\bm{ \kappa}=\diag(\kappa_1, \ldots, \kappa_N)^{\top}$, $\kappa_i= {{\mu}_{R}}_{i}\mu_{i}^{-1}$.
\label{remark2}
\end{remark}

Using (\ref{target-1}) and (\ref{target-2}), it is easier to impose the condition: the eigenvalues of $(\bm \alpha+\bm \lambda \mathbf {D}^{-1}\mathbf {A })\bm \kappa+\bm \beta$ must be less than one in modulus. By the arguments of block triangular and Gershorin circle theorem, $\rho ((\bm \alpha+\bm \lambda \mathbf {D}^{-1}\mathbf {A })\bm \kappa+\bm \beta)\leq \alpha+\lambda \kappa_{(N)}+\beta<1$, where $\kappa_{(N)}$ is the maximum of $\kappa_{i}$ for $i=1,\ldots, N$.

 It is worth mentioning that the ``targeting parameterisation" can make use of the estimations of $\bm {\mu}_{R}$, $\bm{ \mu}$ and $\bm {\kappa}$, where
\begin{align*}
\widehat{ {\mu}}_{i}=\frac{1}{T}\sum_{t=1}^{T} {r}_{it}^2,\quad \quad	\widehat{ {\mu}}_{Ri}=\frac{1}{T}\sum_{t=1}^{T}{\rm{RM}}_{it},   \quad \quad \widehat{ \kappa}_{i}=\widehat{ \mu}_{Ri}\widehat{ \mu}_i^{-1}.
\end{align*}
When these estimators are plugged into the quasi-likelihood functions, the optimization tasks can be substantially simplified, but it does alter the resulting asymptotic standard errors. 

\section{Quasi-maximum likelihood Estimation}

Assume that returns $(\bm{r}_1^2, \ldots, \bm{r}_T^2)$ and the related realized measures $(\mathbf{RM}_1, \ldots, \mathbf{RM}_T)$ from the model (\ref{VARMA}). Let $\phi_0=(\omega_0, \alpha_0, \lambda_0, \beta_0)^{\top}$, ${\phi_R}_0=({\omega_R}_0, {\alpha_R}_0,  {\lambda_R}_0, {\beta_R}_0)^{\top}$. Denote $\theta_0=( \phi_0^{\top}, {\phi_{R}}_0^{\top})^{\top}\in \mathbb{R}^{8}$  is the true value, $\theta=( \phi^{\top}, \phi_{R}^{\top})^{\top}\in \mathbb{R}^{8}$ is the parmeter.
When modeling inference, we will regard the parameters as having no link between the NHEAVY-r and NHEAVY-RM models, i.e. $\phi^{\top}$ and $ \phi_R^{\top}$ are variation free, which is important in the one-step estimation for the NHEAVY model. 

The quasi-log-likelihood function (ignoring the constants) of returns $\{r_{it}^{2}\}$ is given by 
\begin{align}
	\widehat {L}^{r}(\phi)=\frac{1}{T}\sum_{t=2}^{T}\widehat {\ell}_{t}^{r}(\phi) \quad {\rm with} \quad \widehat {\ell}_{t}^{r}(\phi)=\frac{1}{N}\sum_{i=1}^{N}\left\{\log \widehat {h}_{it}(\phi)+\frac{r_{it}^{2}} {\widehat {h}_{it}(\phi)}\right\},
	\label{likehood-return}
\end{align}
where $\widehat {h}_{it}(\phi)$ is defined recursively for $t>1$ by 
\begin{align*}
	 \widehat {h}_{it}(\phi)=\omega+\alpha {\rm RM}_{it-1}+\lambda d_{i}^{-1}\sum_{j\neq i}a_{ij}{\rm RM}_{jt-1}+\beta \widehat {h}_{it-1}(\phi),
	\end{align*}
with $\widehat {h}_{i1}(\phi)=T^{-1/2}\sum_{t=1}^{(\sqrt{T})}r_{it}^2$.

The quasi-log-likelihood function (ignoring the constants) of realized measures $\{{\rm RM}_{it}\}$ is given by 
\begin{align}
	\widehat {L}^{\rm RM}(\phi_R)=\frac{1}{T}\sum_{t=2}^{T}\widehat {\ell}_{t}^{\rm RM}(\phi_R) \quad {\rm with} \quad \widehat {\ell}_{t}^{\rm RM}(\phi_R)=\frac{1}{N}\sum_{i=1}^{N}\left\{\log \widehat{\mu}_{it}(\phi_R)+\frac{{\rm RM}_{it}}{\widehat{\mu}_{it}(\phi_R)}\right\},
	\label{likelihood-RM}
\end{align}
where $\widehat{\mu}_{it}(\phi_R)$ is defined recursively for $t>1$ by 
\begin{align*}
\widehat{\mu}_{it}(\phi_R)=\omega_{R}+\alpha_{R} {\rm RM}_{it-1}+\lambda_{R} d_{i}^{-1}\sum_{j\neq i}a_{ij}{\rm RM}_{jt-1}+\beta_{R} \widehat {\mu}_{it-1}(\phi_R),
\end{align*}
with $\widehat {\mu}_{i1}(\phi_R)=T^{-1/2}\sum_{t=1}^{(\sqrt{T})}{\rm RM}_{it}$.

Then the QMLE of $\phi$ is defined as 
\begin{align*}
	\widehat{\phi}=(\widehat{\omega}, \widehat{\alpha}, \widehat{\lambda}, \widehat{\beta})^{\top}=\arg \min_{\phi\in \Phi} 	\widehat {L}^{r}(\phi),
\end{align*}
where $\Phi$ is the parameter space of $\phi$.
And the QMLE of $\phi_R$ is defined as 
\begin{align*}
	\widehat{\phi}_R=(\widehat{\omega}_R, \widehat{\alpha}_R,  \widehat{\lambda}_R, \widehat{\beta}_R)^{\top}=\arg \min_{\phi_R\in \Phi_R} 	\widehat {L}^{\rm RM}(\phi_R),
\end{align*}
where $\Phi_R$ is the parameter space of $\phi_R$. 

Define $\widehat {\theta} =(\widehat{\phi}^{\top}, \widehat{\phi}_R^{\top})^{\top}$.
To discuss the asymptotic properties of $\widehat {\theta}$, we fist study the asymptotic properties of $\widehat{\phi}$ and $\widehat{\phi}_R$. It is convenient to approximate the sequences $\{ \widehat{h}_{it}(\phi)\}$ and $\{\widehat{\mu}_{it}(\phi_R)\}$ by the ergodic stationary sequence $\{ {h}_{it}(\phi)\}$ and $\{{\mu}_{it}(\phi_R)\}$ as follows,
\begin{align*}
	 {h}_{it}(\phi)=&\omega+\alpha {\rm RM}_{it-1}+\lambda d_{i}^{-1}\sum_{j\neq i}a_{ij}{\rm RM}_{jt-1}+\beta  {h}_{it-1}(\phi),\\
	{\mu}_{it}(\phi_R)=&\omega_{R}+\alpha_{R} {\rm RM}_{it-1}+\lambda_{R} d_{i}^{-1}\sum_{j\neq i}a_{ij}{\rm RM}_{jt-1}+\beta_{R}{\mu}_{it-1}(\phi_R),
\end{align*}
for any $t$ and each $i$. Similarly to the definition of $\widehat {L}^{r}(\phi)$, $\widehat {l}_{t}^{r}(\phi)$, $	\widehat {L}^{\rm RM}(\phi_R)$, and $\widehat {l}_{t}^{\rm RM}(\phi_R)$, we can define
\begin{align}
	&{L}^{r}(\phi)=\frac{1}{T}\sum_{t=2}^{T} {\ell}_{t}^{r}(\phi) \quad {\rm with} \quad  {\ell}_{t}^{r}(\phi)=\frac{1}{N}\sum_{i=1}^{N}\left\{\log  {h}_{it}(\phi)+\frac{r_{it}^{2}} { {h}_{it}(\phi)}\right\},\label{likehood-return-true}\\
	 &{L}^{\rm RM}(\phi_R)=\frac{1}{T}\sum_{t=2}^{T} {\ell}_{t}^{\rm RM}(\phi_R) \quad {\rm with} \quad  {\ell}_{t}^{\rm RM}(\phi_R)=\frac{1}{N}\sum_{i=1}^{N}\left\{\log {\mu}_{it}(\phi_R)+\frac{{\rm RM}_{it}}{{\mu}_{it}(\phi_R)}\right\},
	\label{likelihood-RM-true}
\end{align}

Before stating our main results, we give the following assumptions that are standard in studying quasi-maximum likelihood estimation.
\begin{assum}
	The parameter space $\Theta$ is a compact subset of $\{\theta: \omega >0, \alpha>0, \lambda >0, \omega_{R} >0, \alpha_{R}>0, \lambda_{R}>0, \beta_{R}>0,  0<\beta <1, \alpha+\lambda \kappa_{(N)}+\beta<1, \alpha_{R}+\lambda_{R}+\beta_{R}<1\}$ and $\theta_{0}\in \Theta$.
	\label{assumption1}
\end{assum}
\begin{assum}
	\label{assumption2}
Both $\{{u}_{it}\}$ and $\{{{u}_{Rt}}_{i}\}$ in the model (\ref{VARMA}),  are i.i.d. across $i$ and $t$ with zero mean and finite variance. Moreover, assume $\{{u}_{it}\}$ and $\{{{u}_{Rt}}_{i}\}$ are nondegenerate.
\end{assum}
The strong consistency and asymptotic normality of the QMLE $\widehat{\theta}$ for the NHEAVY model can be stated in the following theorem.
\begin{thm}
If Assumption \ref{assumption1} and  \ref{assumption2} hold,  then $\widehat{\theta}\rightarrow \theta_{0}$  in almost surely as $T \rightarrow \infty$.  Furthermore, if $\theta_{0}$ is an interior point of $ \Theta$, we have
\begin{align}
	\sqrt{NT}\left(\widehat{\theta}-\theta_{0}\right) \stackrel{d}{\longrightarrow} N\left(0, \mathcal{I}^{-1} \mathcal{J}(\mathcal{I}^{-1})^{\top}\right)
	\label{one-step-estimation}
\end{align}
where
\begin{align*}
\mathcal{I}= \begin{pmatrix}
		&{\bm\Sigma_{r}}& \mathbf{0}\\
		&\mathbf{0} & {\bm\Sigma_{R}}
	\end{pmatrix},
	\quad
	 \mathcal{J}=\begin{pmatrix}
		&\left(\kappa_{2}^r-1\right) {\bm\Sigma_{r}}& \left(\kappa_{2}^{r,R}-1\right) {\bm\Sigma^{r,R}}\\
		& \left(\kappa_{2}^{r,R}-1\right) {\bm\Sigma^{R,r}}&(\kappa_{2}^{R}-1) {\bm\Sigma_{R}}
	\end{pmatrix}.
\end{align*}
where $\bm\Sigma^r=\frac{1}{N} \sum_{i=1}^{N} {\mathrm E}\left(\frac{1}{h_{i t}^2(\phi_0)} \frac{\partial h_{i t}\left(\phi_{0}\right)}{\partial \phi} \frac{\partial h_{i t}\left(\phi_{0}\right)}{\partial \phi^{\top}}\right)$, $\bm\Sigma^{r,R}=\frac{1}{N} \sum_{i=1}^{N} {\mathrm E}\left(\frac{1}{h_{i t}(\phi_0)u_{i t}({\phi_R}_0)} \frac{\partial h_{i t}\left(\phi_{0}\right)}{\partial \phi} \frac{\partial u_{i t}\left({\phi_R}_{0}\right)}{\partial \phi_{R}^{\top}}\right)$, $\bm\Sigma^R=\frac{1}{N} \sum_{i=1}^{N} {\mathrm E}\left(\frac{1}{u_{i t}^2({\phi_R}_0)} \frac{\partial u_{i t}\left({\phi_R}_{0}\right)}{\partial \phi_R} \frac{\partial u_{i t}\left({\phi_R}_{0}\right)}{\partial \phi_{R}^{\top}}\right)$,  $\bm\Sigma^{R,r}=\frac{1}{N} \sum_{i=1}^{N} {\mathrm E}\left(\frac{1}{h_{i t}(\phi_0)u_{i t}({\phi_R}_0)}  \frac{\partial u_{i t}\left({\phi_R}_{0}\right)}{\partial \phi_{R}}\frac{\partial h_{i t}\left(\phi_{0}\right)}{\partial \phi^{\top}}\right)$,
$\kappa_{2}^{r}={\mathrm E} (\varepsilon_{i t}^2)$, $\kappa_{2}^{R}={\mathrm E} (\epsilon_{i t}^2)$, and $\kappa_{2}^{r,R}={\mathrm E} (\varepsilon_{i t}\epsilon_{i t})={\mathrm E} (\epsilon_{i t}\varepsilon_{i t})$.
\end{thm}
The block diagonality of (\ref{one-step-estimation}) is due to the variation-free property of parameters, which is a straightforward application of quasi-likelihood theory and can be viewed as an extension of \cite{Bollers:2007}, and is also discussed extensively in \cite{Cipollini:2007}.

Then, we discuss the two-step estimation for the NHEAVY model. With the help of Remark \ref{remark2}, the targeting parameterization can be written in an alternative form,
\begin{equation}
	\begin{split}
	h_{it}=&(1-\alpha\kappa_{i}-\beta)\mu_{i}-\lambda d_{i}^{-1}\sum_{j\neq i}^{N}a_{ij}\mu_{jR}+\alpha {\rm RM}_{it-1}
	+\lambda d_{i}^{-1}\sum_{j\neq i}^{N}a_{ij}{\rm RM}_{jt-1}+\beta h_{it-1}, \\
	\mu_{it}=&(1-\alpha_{R}-\beta_{R})\mu_{iR}-\lambda_{R} d_{i}^{-1}\sum_{j\neq i}^{N}a_{ij}\mu_{jR}+\alpha_{R} {\rm RM}_{it-1}
	+\lambda_{R} d_{i}^{-1}\sum_{j\neq i}^{N}a_{ij}{\rm RM}_{jt-1}+\beta_{R} \mu_{it-1},
	\end{split}
\label{target-form}
\end{equation}
Model (\ref{target-form}) can make a convenient two-step approach. The unconditional moments, $\mu_i$ and ${\mu}_{Ri}$, will be estimated in the first step by 
\begin{align*}
	\widehat{ {\mu}}_{i}=\frac{1}{T}\sum_{t=1}^{T} {r}_{it}^2,\quad \quad	\widehat{ {\mu}}_{Ri}=\frac{1}{T}\sum_{t=1}^{T}{\rm{RM}}_{it}. 
\end{align*}
Defined $ \widehat{l}^{r}( \bm{\mu}, \bm{\mu}_{R}, \widehat {\phi})$ and $\widehat{l}^{\rm RM}( \bm{\mu}_{R}, \widehat{\phi}_{R})$ to be the observed  log-likelihood function for the covariance targeting HEAVY model. Then $\widehat {\phi}^{\top}$ and $\widehat{\phi}_{R}^{\top}$ will be estimated by QMLE in the second step.
\begin{align*}
	\widehat{\phi}=\mathop{\arg\min}_{\theta\in \Theta}\widehat {L}^{r}(\bm\mu, \bm\mu_{R}, \phi) \quad {\rm and} \quad \widehat{\phi}_R=\mathop{\arg\min}_{\theta_R \in \Theta_R}\widehat {L}^{\rm RM}( \bm\mu_{R}, \phi_R).
\end{align*}
With ``targeting parameterisation" in Remark \ref{remark2}, variation freeness between the parameters of the NHEAVY-r and NHEAVY-RM equations no longer holds since $\bm {\kappa}$ depends on ${\bm \mu}_R$.
And the asymptotic properties of parameter estimators refer to HAC estimator \citep{Noureldin:2012}.

\section{Simulation studies}

To evaluate the finite-sample performance of the proposed mothodology, we give three simulation experiments. For a given network structure $\mathbf {A}$ used in the experiments, the true high-frequency log price $\mathbf {X}(t_{l,m})=(X_{1}(t_{l,m}),\ldots, X_{N}(t_{l,m}))^{T}$ at $t_{{l,m}}=l-1+m/M, l=1,\ldots, L, m=1, \cdots, M$ is generated from the following model,
\begin{equation}
	d \mathbf {X} (t) = \bm {\mu}_t dt + \bm {\sigma}_t ^{\top}  d\mathbf {B}_t\text{,}
	\label{model 4.1}
\end{equation}
where ${\bm {\mu}}_t \in \mathbb {R}^{N}$ is the drift, $\bm {\sigma}_t\in \mathbb {R}^{N\times N}$ with $\bm {\gamma}(t)=\bm {\sigma}_t^{T}\bm {\sigma}_t $ being the volatility matrix of $X(t)$, $\mathbf {B}_t \in \mathbb {R}^{N}$ is a standard  $N$-dimensional Brownian motion. For convernience, let $\bm{ \mu}_t=0$ and the volatility $\bm{ \sigma}_t$ have a Cholesky decomposition,
\begin{equation*}
	\bm{ \gamma}(t)=(\gamma_{ij}(t)), \quad \gamma_{ij}(t)=\sqrt{\tau_i \tau_j }\kappa^{|i-j|},
\end{equation*}
where $\{\tau_i, i=1,\ldots,N\}$ are independently generated from a uniform distribution on $[0,1]$, and $\kappa$ is set at 0.5. The noisy observations $\{Y_{i}(t_{l}), i=1, \cdots, N\}$ from model (\ref{model 4.1}) are added mean zero normal random errors $\{\xi_{i}(t_{l}), i=1,\cdots, N\}$. In this noise suitation, we considered MSRV as the estimators of daily volatility following \cite{KW:2016}.

\textbf{Example 1} (Dyad Independence Model) A dyad is defined as $A_{i,j}=(a_{ij}, a_{ji})$ for any $1\leq i < j\leq N.$ Dyad independence assumes that different $A_{i,j}$ 's are independent. In order to ensure network sparsity, we set $P(A_{i,j}=(1,1))=20N^{-1}$. As a result, the expected number of mutually connected dyads is of $\mathcal{O}(N)$. Next, set $P(A_{ij}=(1,0))=P(A_{ij}=(0,1))=0.5N^{-0.8}$. The histograms of the out-degree and the in-degree of this network structure are shown in Figure \ref{fig_1}.
Fix $N=25, 50$ \rm or $100$, $T=50$ \rm or $100$ and $m=39, 78$ \rm  or $390$. The initial estimator for one-step method are $\phi_0=(0.005, 0.001,0.001,0.9)$ and ${\phi_R}_0=(0.005, 0.1,0.1,0.5)$,  and the initial estimator  for two -step method are $\phi_0=(0.01,0.001,0.7)$ and ${\phi_R}_0=(0.001, 0.01, 0.85)$ . 
All experiment are with 1000 replications, that is $Q=1000$. To evaluate the proposed estimators, the root-mean-square error is calculated by ${\rm RMSE}=\{Q^{-1}\sum\limits_{q=1}^{Q}(\widehat{\theta}_{i}^{(q)}-\theta_{i}^{(q)})^{2}\}^{1/2}, i=1,2,3,4 $ or $i=1,2,3$. And the network density (ND), i.e. ${N(N-1)}^{-1}\sum_{i,j}a_{i,j}$, is also reported.
The results are summaried in Tables \ref{simulation_Dayone} and \ref{simulation_Daytwo}.

\textbf{Example 2} (Power-Law Distribution Model) In the existing literature, a power-law distribution reflects a popular network phenomenon, that is, the majority of nodes have very few edges but a small amount has a huge number of edges. To mimic this phenomenon, we simulate adjacency matrix $\mathbf {A}$ as fellows. First, for each node, its in-degree satisfies $P(d_i=k)=ck^{-\alpha}$ for a normalizing constant $c$ and exponent parameter $\alpha\in{1,2,3}$. A smaller $\alpha$ value implies a heavier tail. Next, for the $i$-th node, we randomly select $d_i$ nodes to be its followers. $N, T$ and $m$ are set at same values in \textbf{Example 1}.  The histograms of the out-degree and the in-degree of this network structure are shown in Figure \ref{fig_1}.
The initial estimators of one-step method are $\phi_0=(0.0006, 0.1,0.005,0.8)$, ${\phi_R}_0=(0.05, 0.1,0.1,0.5)$, and the initial estimators  for two-step method are $\phi_0=(0.0015,0.01,0.72)$ and ${\phi_{R}}_{0}=(0.01,0.015,0.6)$.  The RMSE and ND are summaried in Tables \ref{simulation_Powerone} and  \ref{simulation_Powertwo}.

\textbf{Example 3} (Stochastic Block Model) We next consider another popular network structure, the stochastic block model. This model is of particular interest for community detection. Specially, we assign to a block label $(k=1,\cdots, K)$ with equal probability, where $K=5,10$ \rm or $20$ is the total number of blocks. Next, set $P(a_{ij}=1)=0.3N^{-0.3}$ if $i$ and $j$ belong to the same block, and $P(a_{ij}=1)=0.3N^{-1}$ otherwise. Accordingly, the nodes within the same block are more likely to be connected, compared with nodes from different blocks. $N, T$ and $m$ are set at the same values in \textbf{Examples 1} and \textbf{2}. The histograms of the out-degree and the in-degree of this network structure are presented in Figure \ref{fig_1}.
The initial estimators of one-step method are $\phi_0=(0.01, 0.001,0.02,0.8)$, ${\phi_R}_0=(0.006, 0.1,0.1,0.5)$, and the initial estimators of two-step method are $\phi_0=(0.005,0.01,0.75)$ and ${\phi_{R}}_0=(0.001,0.015,0.85)$ . The RMSE and ND are summaried in Tables \ref{simulation_SBone} and \ref{simulation_SBtwo}.

\begin{figure}[H]
	\centering
		\begin{minipage}[t]{9cm}
		\includegraphics[angle=0,width=9cm]{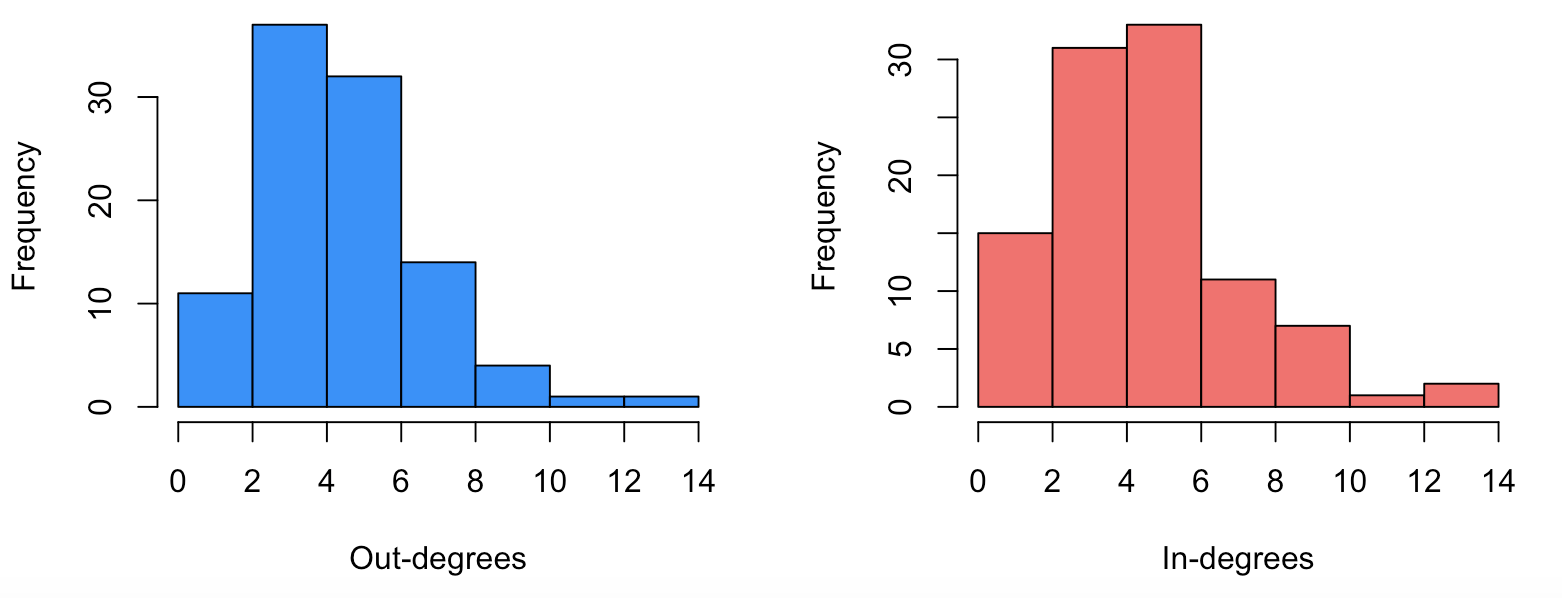}
	\end{minipage}
	\centering
	\begin{minipage}[t]{9cm}
		\includegraphics[angle=0,width=9cm]{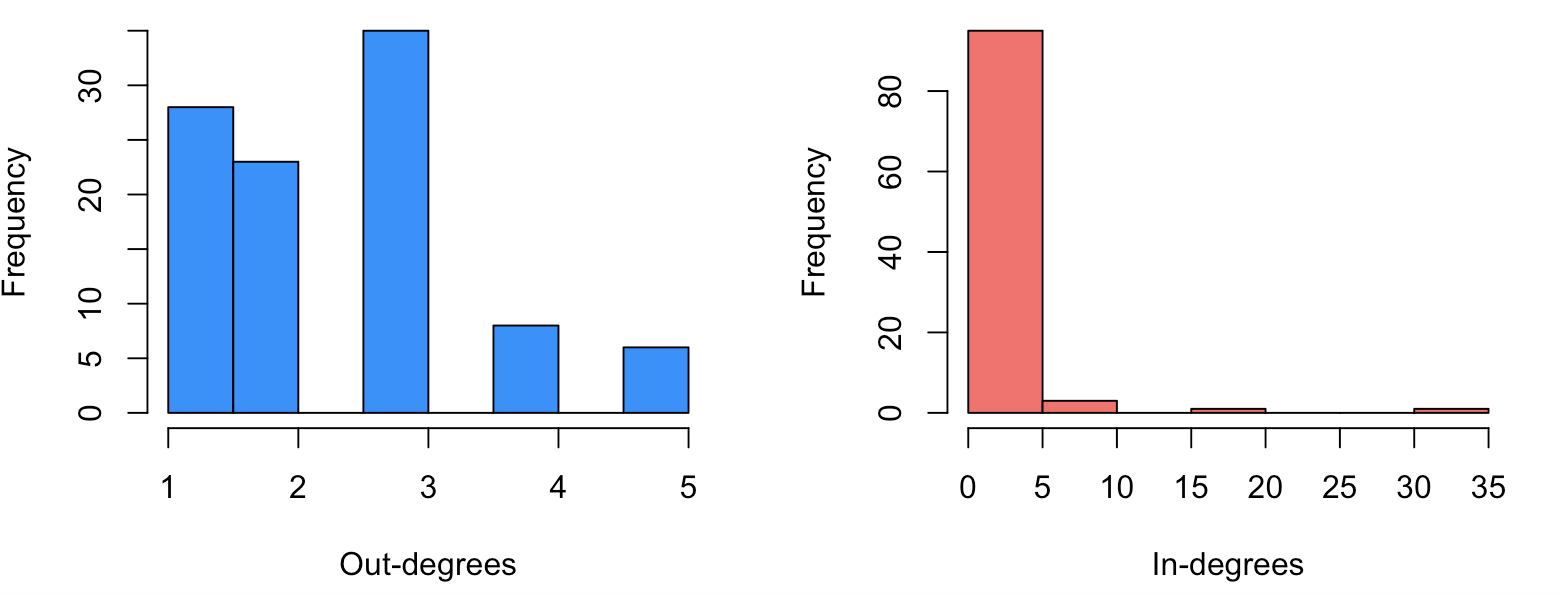}
	\end{minipage}
		\begin{minipage}[t]{9cm}
		\includegraphics[angle=0,width=9cm]{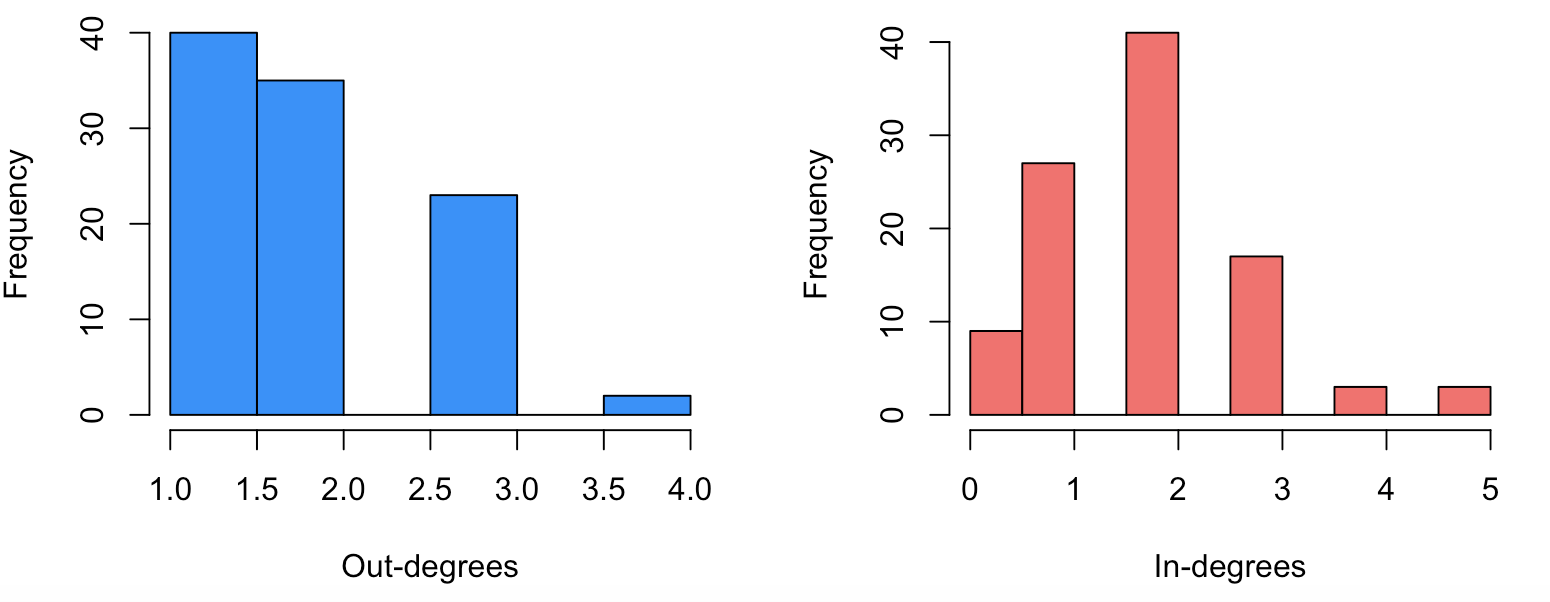}
	\end{minipage}
	\caption{Histograms of the out-degree  (the left plots) and  in-degree (the right plots), the top panel is for Dyad Network structure, the middle panel is for Power-Law Network structure, and the low panel is for Stochastic Block Network structure.}
	\label{fig_1}
\end{figure}
From Tables \ref{simulation_Dayone}- \ref{simulation_SBtwo}, all estimated parameters  are consistent, and RMSEs decease toward zero as $T \rightarrow \infty$.
This may be due to the reason that, when $T$ is larger, more sampled time intervals are available to fit the model. However, for our proposed model, we also need to compute the realized volatility estimators over each time interval for further model estimation. Less transactions are obtained over shorter time intervals, and this may render more unreliable RV estimators, which are used for the following Quasi-likelihood optimization.  In sum, these findings confirm that the proposed estimator $\hat{\theta}$ is consistent.
\begin{table}[H]
	\centering
	\scalebox{0.7}{
	\begin{tabular}{c|ccc|ccc|cccc}
		\hline
		\hline
		\multirow{2}{*}{  } & \multicolumn{3}{c} {$    T=50, m=39$ } & \multicolumn{3}{c} {$    T=50, m=78$ }& \multicolumn{3}{c} {$    T=50, m=390$ }\\
		\cline{2-10}
		& $N=25$ & $N=50$ & $ N=100$ & $ N=25$ & $N=50$ & $N=100$ & $N=25$ & $ N=50$ & $N=100$ &  \\
		$\omega_{0}$  &6.030e-04  & 5.976e-04&5.970e-04 &6.002e-04&5.981e-04&5.968e-04&5.991e-04&5.991e-04&5.981e-04\\
		$\alpha_{0}$  &1.833e-02  &1.151e-02 & 9.207e-03&1.599e-02&1.158e-02&7.899e-03&7.244e-04&4.604e-04&3.645e-04\\
		$\lambda_{0}$  &0.100 &6.817e-02&5.578e-02&7.678e-02&4.893e-02&3.026e-02&5.868e-04&4.076e-04&3.986e-04\\
		$\beta_{0}$  & 3.052e-02  &1.568e-02&5.353e-03&1.569e-02&7.235e-03&8.984e-04&2.310e-06&1.791e-06&1.736e-06\\
		$\omega_{R}$  &4.965e-03 &4.973e-03&4.988e-03&4.977e-03&4.981e-03&4.983e-03&4.996e-03&4.995e-03&4.994e-03\\
		$\alpha_{R}$  &0.105  &0.107&0.912e-02&9.521e-02&9.034e-02&9.253e-02&7.835e-02&7.969e-02&6.276e-02\\
		$\lambda_{R}$  & 0.117  & 0.101& 9.400e-02&0.117&0.126&6.306e-02&5.309e-02&6.231e-02&9.773e-03\\
		$\beta_{R}$  & 0.122 &0.118&0.115&9.764e-02&8.885e-02&7.146e-02&3.898e-02&4.030e-02&2.074e-02\\
		\hline
		\multirow{2}{*}{  } & \multicolumn{3}{c} {$    T=100, m=39$ } & \multicolumn{3}{c} {$    T=100, m=78$ }& \multicolumn{3}{c} {$    T=100, m=390$ }\\
		\cline{2-10}
		& $N=25$ & $N=50$ & $ N=100$ & $ N=25$ & $N=50$ & $N=100$ & $N=25$ & $ N=50$ & $N=100$ &  \\
		$\omega_{0}$  & 5.979e-04  & 5.962e-04& 5.958e-04&5.970e-04&5.948e-04&5.948e-04&5.949e-04&5.935e-04&5.925e-04\\
		$\alpha_{0}$  &8.873e-03   &6.159e-03 &3.615e-03 &4.978e-03&3.279e-03&1.933e-03&4.694e-05&3.443e-05&2.987e-05\\
		$\lambda_{0}$  & 3.488e-02  &1.577e-02&7.279e-03&5.260e-03&3.310e-03&1.886e-03&4.372e-05&3.243e-05&2.938e-05\\
		$\beta_{0}$  &  5.905e-03 &7.570e-04&1.072e-04&3.232e-05&1.380e-05&8.446e-06&2.318e-07&1.574e-07&1.456e-07\\
		$\omega_{R}$  & 4.975e-03  &4.974e-03&4.974e-03&4.974e-03&4.964e-03&4.863e-03&4.989e-03&4.898e-03&4.948e-03\\
		$\alpha_{R}$  & 0.110  &0.106&0.102&9.509e-02&9.492e-02&9.279e-02&7.352e-02&7.341e-02&5.991e-02\\
		$\lambda_{R}$  & 0.118  & 0.109& 9.912e-02&7.794e-02&6.641e-02&1.617e-02&5.669e-02&5.645e-02&5.219e-03\\
		$\beta_{R}$  & 9.056e-02 &7.055e-02&5.110e-02&0.110&8.467e-02&5.130e-02&2.185e-02&2.073e-02&1.859e-02\\
		\hline
		ND(\%) &  20.8 &10.2&5.0&  20.8 &10.2&5.0&  20.8 &10.2&5.0\\
		\hline
		\hline
	\end{tabular}
}
	\caption{Simulation results for \textbf{Example 1} with 1000 replications for one-step NHEAVY models. The RMSE and ND(\%) are reported for each estimator. }
	\label{simulation_Dayone}
\end{table}

\begin{table}[H]
	\centering
		\scalebox{0.7}{
	\begin{tabular}{c|ccc|ccc|cccc}
		\hline
		\hline
		\multirow{2}{*}{  } & \multicolumn{3}{c} {$    T=50, m=39$ } & \multicolumn{3}{c} {$    T=50, m=78$ }& \multicolumn{3}{c} {$    T=50, m=390$ }\\
		\cline{2-10}
		& $N=25$ & $N=50$ & $ N=100$ & $ N=25$ & $N=50$ & $N=100$ & $N=25$ & $ N=50$ & $N=100$ &  \\
		$\alpha_{0}$  &2.378e-02  &1.808e-02 & 1.497e-02&2.285e-02&1.930e-02&1.536e-02&3.472e-02&2.416e-02&1.423e-02\\
		$\lambda_{0}$  &0.113 &7.240e-02&7.144e-02&0.112&0.111&6.943e-02&8.762e-02&7.450e-02&6.656e-02\\
		$\beta_{0}$  & 3.568e-02  &2.101e-02&1.427e-02&3.380e-02&1.999e-02&1.444e-02&3.149e-02&2.096e-02&1.448e-02\\
		$\alpha_{R}$  &  7.588e-02&7.251e-02&7.074e-02&7.692e-02&7.263e-02&7.135e-02&3.956e-02&3.540e-02&3.586e-02\\
		$\lambda_{R}$  & 9.355e-02  & 8.687e-02& 9.351e-02&7.806e-02&5.619e-02&5.584e-02&8.058e-02&5.617e-02&3.544e-02\\
		$\beta_{R}$  & 8.186e-02 &7.987e-02&7.910e-02&8.213e-02&8.037e-02&7.978e-02&2.796e-02&1.972e-02&1.322e-02\\
		\hline
		\multirow{2}{*}{  } & \multicolumn{3}{c} {$    T=100, m=39$ } & \multicolumn{3}{c} {$    T=100, m=78$ }& \multicolumn{3}{c} {$    T=100, m=390$ }\\
		\cline{2-10}
		& $N=25$ & $N=50$ & $ N=100$ & $ N=25$ & $N=50$ & $N=100$ & $N=25$ & $ N=50$ & $N=100$ &  \\
		$\alpha_{0}$  & 1.387e-02  &1.374e-02  &1.222e-02 &1.369e-02&1.365e-02&1.207e-02&1.376e-02&1.348e-02&1.056e-02\\
		$\lambda_{0}$  &7.046e-02   &4.709e-02&3.150e-02&7.025e-02&4.654e-02&3.149e-02&5.478e-02&4.719e-02&3.143e-02\\
		$\beta_{0}$  & 2.612e-02  &1.627e-02&1.488e-02&2.550e-02&1.586e-02&1.141e-02&2.166e-02&1.541e-02&1.075e-02\\
		$\alpha_{R}$  & 3.760e-02  &3.502e-02&3.425e-02&3.744e-02&3.496e-02&3.412e-02&3.596e-02&3.582e-02&3.394e-02\\
		$\lambda_{R}$  & 4.557e-02  &3.284e-02& 2.473e-02&4.550e-02&3.263e-02&2.362e-02&4.088e-02&3.283e-02&2.094e-02\\
		$\beta_{R}$  & 2.461e-02  &1.598e-02&1.186e-02&2.459e-02&1.668e-02&1.173e-02&2.153e-02&1.652e-02&1.159e-02\\
		\hline
		ND(\%) &  20.8 &10.2&5.0&  20.8 &10.2&5.0&  20.8 &10.2&5.0\\
		\hline
		\hline
	\end{tabular} 
}
	\caption{Simulation results for \textbf{Example 1} with 1000 replications for two-step NHEAVY models. The RMSE and ND(\%) are reported for each estimator. }
	\label{simulation_Daytwo}
\end{table}

\begin{table}[H]
	\centering
		\scalebox{0.7}{
	\begin{tabular}{c|ccc|ccc|cccc}
		\hline
		\hline
		\multirow{2}{*}{  } & \multicolumn{3}{c} {$    T=50, m=39$ } & \multicolumn{3}{c} {$    T=50, m=78$ }& \multicolumn{3}{c} {$    T=50, m=390$ }\\
		\cline{2-10}
		& $N=25$ & $N=50$ & $ N=100$ & $ N=25$ & $N=50$ & $N=100$ & $N=25$ & $ N=50$ & $N=100$ &  \\
		$\omega_{0}$  & 5.897e-04 & 5.900e-04& 5.852e-04&5.959e-04&5.963e-04&5.987e-04&5.997e-02&5.996e-04&5.996e-04\\
		$\alpha_{0}$  &  2.194e-02&1.576e-02& 1.171e-02&2.630e-02&1.927e-02&1.310e-02&1.344e-02&8.923e-03&6.723e-03\\
		$\lambda_{0}$  &3.160e-02 &2.589e-02&2.560e-02&3.619e-02&2.396e-02&2.386e-02&3.242e-02&2.520e-02&1.937e-02\\
		$\beta_{0}$  & 8.715e-02  &8.476e-02&8.199e-02&3.199e-02&2.371e-02&1.473e-02&1.933e-02&1.372e-02&4.535e-03\\
		$\omega_{R}$  & 4.998e-02&4.998e-02&4.997e-02&4.998e-02&4.998e-02&4.997e-02&4.999e-02&4.996e-02&4.996e-02\\
		$\alpha_{R}$  &0.107  &0.102&9.701e-02&9.806e-02&9.069e-02&9.019e-02&8.095e-02&6.418e-02&9.778e-03\\
		$\lambda_{R}$  & 0.101  & 9.132e-02&6.483e-02 &3.197e-02&2.509e-02&2.000e-02&6.313e-02&4.411e-02&1.552e-02\\
		$\beta_{R}$  &6.762e-02 &4.332e-02&4.048e-02&9.729e-02&7.612e-02&4.948e-02&0.105&9.927e-02&5.664e-02\\
		\hline
		\multirow{2}{*}{  } & \multicolumn{3}{c} {$    T=100, m=39$ } & \multicolumn{3}{c} {$    T=100, m=78$ }& \multicolumn{3}{c} {$    T=100, m=390$ }\\
		\cline{2-10}
		& $N=25$ & $N=50$ & $ N=100$ & $ N=25$ & $N=50$ & $N=100$ & $N=25$ & $ N=50$ & $N=100$ &  \\
		$\omega_{0}$  & 5.996e-04  & 5.959e-04& 5.958e-04&5.986e-04&5.986e-04&5.986e-04&5.996e-04&5.997e-04&5.997e-04\\
		$\alpha_{0}$  & 1.911e-02  &1.641e-02 &9.102e-03 &1.858e-02&1.538e-02&7.937e-02&5.865e-03&5.846e-03&5.863e-03\\
		$\lambda_{0}$  & 2.736e-02      &2.280e-02&2.567e-02&2.010e-02&1.513e-02&1.094e-02&1.629e-03&1.537e-03&1.568e-03\\
		$\beta_{0}$  & 8.034e-02  &7.435e-02&7.322e-02&4.103e-02&3.151e-02&2.843e-02&5.931e-03&6.133e-03&5.913e-03\\
		$\omega_{R}$  &4.996e-02   &4.996e-02&4.996e-02&4.998e-02&5.000e-02&4.997e-02&5.000e-02&4.999e-02&4.997e-02\\
		$\alpha_{R}$  & 9.615e-02  &9.023e-02&8.833e-02&9.494e-02&9.494e-02&9.394e-02&2.172e-02&1.926e-02&1.610e-02\\
		$\lambda_{R}$  &  2.116e-02 & 1.435e-02&1.025e-02&2.114e-02&1.700e-02&1.116e-02&1.296e-02&5.566e-03&5.857e-03\\
		$\beta_{R}$  & 5.977e-02  &3.066e-02&1.920e-02&3.353e-02&2.826e-02&1.741e-02&3.977e-02&2.635e-02&1.167e-02\\
		\hline
		ND(\%) &  7.2 &4.5&2.4&  7.2 &4.5&2.4&  7.2&4.5&2.4\\
		\hline
		\hline
	\end{tabular} 
}
	\caption{Simulation results for \textbf{Example 2} with 1000 replications for one-step NHEAVY models. The RMSE and ND(\%) are reported for each estimator. }
	\label{simulation_Powerone}
\end{table}

\begin{table}[H]
	\centering
	\scalebox{0.7}{
	\begin{tabular}{c|ccc|ccc|cccc}
		\hline
		\hline
		\multirow{2}{*}{  } & \multicolumn{3}{c} {$    T=50, m=39$ } & \multicolumn{3}{c} {$    T=50, m=78$ }& \multicolumn{3}{c} {$    T=50, m=390$ }\\
		\cline{2-10}
		& $N=25$ & $N=50$ & $ N=100$ & $ N=25$ & $N=50$ & $N=100$ & $N=25$ & $ N=50$ & $N=100$ &  \\
		$\alpha_{0}$  &2.117e-02  &1.509e-02 & 1.604e-02&2.368e-02&1.608e-02&1.103e-02&3.554e-02&2.400e-02&1.379e-02\\
		$\lambda_{0}$  & 3.022e-02&2.230e-02&2.823e-02&3.286e-02&2.332e-02&2.073e-02&4.519e-02&3.247e-02&1.556e-02\\
		$\beta_{0}$  & 3.042e-02  &2.014e-02&1.505e-02&3.336e-02&2.065e-02&1.398e-02&3.172e-02&2.041e-02&1.520e-02\\
		$\alpha_{R}$  & 8.558e-02 &8.072e-02&7.983e-02&8.553e-02&8.141e-02&8.085e-02&8.764e-02&8.393e-02&8.115e-02\\
		$\lambda_{R}$  &2.302e-02   & 2.160e-02&2.173e-02 &2.353e-02&2.233e-02&2.215e-02&2.318e-02&2.187e-02&1.876e-02\\
		$\beta_{R}$  & 2.756e-02 &1.872e-02&1.326e-02&3.610e-02&1.934e-02&1.295e-02&2.984e-02&2.052e-02&1.937e-02\\
		\hline
		\multirow{2}{*}{  } & \multicolumn{3}{c} {$    T=100, m=39$ } & \multicolumn{3}{c} {$    T=100, m=78$ }& \multicolumn{3}{c} {$    T=100, m=390$ }\\
		\cline{2-10}
		& $N=25$ & $N=50$ & $ N=100$ & $ N=25$ & $N=50$ & $N=100$ & $N=25$ & $ N=50$ & $N=100$ &  \\
		$\alpha_{0}$  & 1.421e-02  &9.669e-03  &6.919e-03 &1.407e-02&1.025e-02&6.839e-03&1.212e-02&1.020e-02&6.042e-03\\
		$\lambda_{0}$  &1.875e-02   &1.573e-02&1.364e-02&1.778e-02&1.338e-02&1.364e-02&1.753e-02&1.215e-02&1.367e-02\\
		$\beta_{0}$  & 2.455e-02  &1.510e-02&1.108e-02&2.397e-02&1.508e-02&1.064e-02&2.287e-02&1.502e-02&1.073e-02\\
		$\alpha_{R}$  &  4.598e-02 &4.371e-02&4.323e-02&4.572e-02&3.535e-02&4.357e-02&4.634e-02&8.340e-03&4.350e-02\\
		$\lambda_{R}$  & 1.988e-02  &1.778e-02& 1.643e-02&1.917e-02&9.505e-03&9.643e-03&1.291e-02&9.041e-03&1.643e-03\\
		$\beta_{R}$  &  2.131e-02 &1.547e-02&1.096e-02&3.229e-02&1.261e-02&1.215e-02&3.141e-02&1.119e-02&1.128e-02\\
		\hline
		ND(\%) &  7.2 &4.5&2.4&  7.2 &4.5&2.4&  7.2 &4.5&2.4\\
		\hline
		\hline
	\end{tabular}
}
	\caption{Simulation results for \textbf{Example 2} with 1000 replications for two-step NHEAVY models. The RMSE and ND(\%) are reported for each estimator. }
	\label{simulation_Powertwo}
\end{table}

\begin{table}[H]
	\centering
	\scalebox{0.7}{
	\begin{tabular}{c|ccc|ccc|cccc}
		\hline
		\hline
		\multirow{2}{*}{  } & \multicolumn{3}{c} {$    T=50, m=39$ } & \multicolumn{3}{c} {$    T=50, m=78$ }& \multicolumn{3}{c} {$    T=50, m=390$ }\\
		\cline{2-10}
		& $N=25$ & $N=50$ & $ N=100$ & $ N=25$ & $N=50$ & $N=100$ & $N=25$ & $ N=50$ & $N=100$ &  \\
		$\omega_{0}$  &9.989e-03  & 9.989e-03&9.993e-03 &9.995e-03&9.994e-03&9.993e-03&9.998e-03&9.998e-03&9.996e-03\\
		$\alpha_{0}$  &2.450e-02  &1.803e-02&1.405e-02 &2.277e-02&2.295e-02&1.242e-02&1.266e-02&8.418e-03&6.223e-03\\
		$\lambda_{0}$  & 2.941e-02&2.545e-02&2.372e-02&3.223e-02&2.466e-02&2.222e-02&1.293e-02&1.127e-02&7.522e-03\\
		$\beta_{0}$  &8.585e-02   &8.388e-02&7.263e-02&3.967e-02&3.745e-02&2.146e-02&1.672e-03&1.661e-03&1.069e-03\\
		$\omega_{R}$  &5.977e-03 &5.980e-03&5.988e-03&5.976e-03&5.975e-03&5.974e-03&5.995e-03&5.995e-03&5.990e-03\\
		$\alpha_{R}$  &0.104  &0.103&9.822e-02&0.101&9.076e-02&9.069e-02&7.719e-02&6.826e-02&2.761e-02\\
		$\lambda_{R}$  &5.811e-02&5.607e-02&4.441e-02 &3.255e-02&2.270e-02&2.093e-02&2.514e-02  & 1.965e-02&2.197e-02\\
		$\beta_{R}$  &9.792e-02&8.974e-02&8.438e-02&9.528e-02&7.342e-02&4.534e-02&6.761e-02  &3.987e-02&2.163e-02\\
		\hline
		\multirow{2}{*}{  } & \multicolumn{3}{c} {$    T=100, m=39$ } & \multicolumn{3}{c} {$    T=100, m=78$ }& \multicolumn{3}{c} {$    T=100, m=390$ }\\
		\cline{2-10}
		& $N=25$ & $N=50$ & $ N=100$ & $ N=25$ & $N=50$ & $N=100$ & $N=25$ & $ N=50$ & $N=100$ &  \\
		$\omega_{0}$  & 1.001e-02  & 9.994e-03& 9.993e-03&9.996e-03&9.993e-03&9.991e-03&9.993e-03&9.991e-03&9.991e-03\\
		$\alpha_{0}$  & 2.288e-02  &1.728e-02 &1.386e-02 &1.406e-02&7.931e-03&3.659e-03&1.125e-03&6.042e-04&2.079e-06\\
		$\lambda_{0}$  & 2.587e-02  &1.995e-02&1.472e-02&1.248e-02&5.621e-03&1.805e-03&1.036e-03&5.440e-04&1.865e-06\\
		$\beta_{0}$  & 6.148e-02  &5.118e-02&4.130e-02&1.901e-02&7.332e-03&1.655e-03&4.904e-05&2.851e-05&1.548e-05\\
		$\omega_{R}$  & 5.964e-03  &5.963e-03&5.962e-03 &5.980e-03&5.960e-03&5.960e-03&5.958e-03&5.568e-03&5.559e-03\\
		$\alpha_{R}$  & 9.895e-02  &9.257e-02&8.930e-02&9.793e-02&9.246e-02&8.447e-03&3.745e-03&3.129e-03&4.382e-03\\
		$\lambda_{R}$  & 2.162e-02  & 1.404e-02&1.116e-02&2.148e-02&1.115e-02&1.113e-02&6.837e-04&3.410e-04&1.131e-04\\
		$\beta_{R}$  & 6.886e-02  &3.299e-02&2.073e-02&3.333e-02&3.212e-02&1.190e-02&4.404e-03&4.029e-03&6.059e-03\\
		\hline
		ND(\%) &  7.2 &4.3&1.9&  7.2 &4.3&1.9&  7.2&4.3&1.9\\
		\hline
		\hline
	\end{tabular} 
}
	\caption{Simulation results for \textbf{Example 3} with 1000 replications for one-step NHEAVY models. The RMSE and ND(\%) are reported for each estimator. }
	\label{simulation_SBone}
\end{table}

\begin{table}[H]
	\centering
		\scalebox{0.7}{
	\begin{tabular}{c|ccc|ccc|cccc}
		\hline
		\hline
		\multirow{2}{*}{  } & \multicolumn{3}{c} {$    T=50, m=39$ } & \multicolumn{3}{c} {$    T=50, m=78$ }& \multicolumn{3}{c} {$    T=50, m=390$ }\\
		\cline{2-10}
		& $N=25$ & $N=50$ & $ N=100$ & $ N=25$ & $N=50$ & $N=100$ & $N=25$ & $ N=50$ & $N=100$ &  \\
		$\alpha_{0}$  &2.447e-02  &1.603e-02 &1.203e-02 &2.291e-02&1.722e-02&1.271e-02&3.335e-02&2.363e-02&1.451e-02\\
		$\lambda_{0}$  & 2.852e-02&2.220e-02&2.720e-02&2.814e-02&2.206e-02&2.077e-02&3.520e-02&2.896e-02&2.551e-02\\
		$\beta_{0}$  & 9.526e-02  &9.353e-02&9.489e-02&9.504e-02&9.414e-02&9.385e-02&9.573e-02&9.546e-02&9.446e-02\\
		$\alpha_{R}$  & 7.602e-02 &7.261e-02&7.010e-02&7.724e-02&7.343e-02&7.174e-02&7.885e-02&7.568e-02&7.263e-02\\
		$\lambda_{R}$  & 2.494e-02  & 2.078e-02& 2.051e-02&2.607e-02&2.095e-02&2.020e-02&2.648e-02&2.121e-02&2.053e-02\\
		$\beta_{R}$  & 9.440e-02  &9.270e-02&9.238e-02&9.361e-02&9.250e-02&9.223e-02&9.438e-02&9.144e-02&9.220e-02\\
		\hline
		\multirow{2}{*}{  } & \multicolumn{3}{c} {$    T=100, m=39$ } & \multicolumn{3}{c} {$    T=100, m=78$ }& \multicolumn{3}{c} {$    T=100, m=390$ }\\
		\cline{2-10}
		& $N=25$ & $N=50$ & $ N=100$ & $ N=25$ & $N=50$ & $N=100$ & $N=25$ & $ N=50$ & $N=100$ &  \\
		$\alpha_{0}$  & 1.510e-02  &1.071e-02  &8.488e-03 &1.507e-02&1.038e-02&8.458e-03&1.222e-02&1.009e-02&8.014e-03\\
		$\lambda_{0}$  & 1.890e-02  &1.510e-02&1.285e-02&1.798e-02&1.508e-02&1.222e-02&1.698e-02&1.410e-02&1.021e-02\\
		$\beta_{0}$  & 2.519e-02  &1.526e-02&1.443e-02&2.358e-02&1.483e-02&1.142e-02&2.328e-02&1.504e-02&1.023e-02\\
		$\alpha_{R}$  & 3.759e-02  &3.527e-02&3.402e-02&3.742e-02&3.524e-02&3.304e-02&3.712e-02&3.522e-02&3.221e-02\\
		$\lambda_{R}$  & 1.906e-02  &1.732e-02& 1.645e-02&1.871e-02&1.697e-02&1.643e-02&1.874e-02&1.614e-02&1.605e-03\\
		$\beta_{R}$  & 8.309e-02 &7.985e-02&7.855e-02&8.348e-02&8.093e-02&7.978e-02&8.459e-02&8.226e-02&8.025e-02\\
		\hline
		ND(\%) &  7.2 &4.3&1.9&  7.2 &4.3&1.9&  7.2 &4.35&1.9\\
		\hline
		\hline
	\end{tabular} 
}
	\caption{Simulation results for \textbf{Example 3} with 1000 replications for two-step NHEAVY models. The RMSE and ND(\%) are reported for each estimator. }
	\label{simulation_SBtwo}
\end{table}

\section{An empirical example}

This section chooses 18 constituent stocks from the CSI 300 to analyze the stock markets by applying our proposed model. The CSI 300 is a capitalization-weighted stock market index, designed to replicate the performance of the top 300 stocks traded on the Shanghai Stock Exchange and the Shenzhen Stock Exchange. It is considered to be a blue-chip index for Mainland China stock exchanges. The 18 constituent stocks take up the larger weight in the CSI 300 index.
The high-frequency historical data is the 1-minute data from January 2, 2018 to December 31, 2019. In general, there are 240 prices within a trading day. All high-frequency closing prices are transformed into log prices $\log (P_{t_{l,m}})$, $t_{{l,m}}=l-1+m/M, l=1,\ldots, L, m=1, \cdots, M$ with $L=487$ and $M=240$, where $L$ is the number of trading days, and $M$ is the number of intraday  data. The in-sample period starts from January 2, 2018 to July 8, 2019, which contains 88080 high-frequency prices and 367 days. The out-of-sample period starts  from July 9, 2019 to December 31, 2019, which contains 28800 high-frequency prices and 120 days. The network structure is described by the adjacency matrix $A=(a_{ij}), i =1,\cdots, 18$,  and $j =1, \cdots, 18$. The network structure is visulized in Figure \ref{network structure}. All the chosen stocks and the industry classifition are found in Table \ref{JointQuant1}.  

\begin{figure}[H]
	\centering
	\includegraphics[angle=270,width=6.5cm]{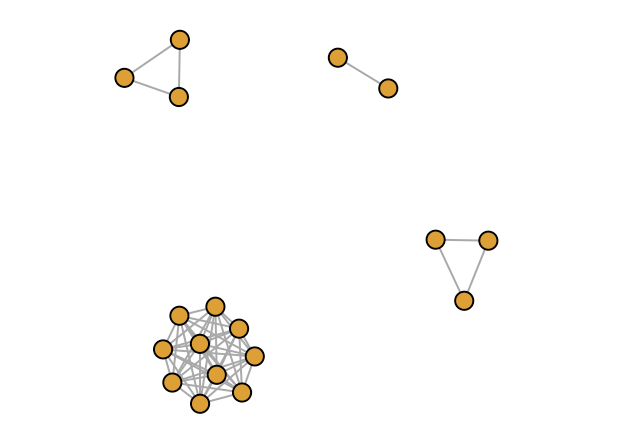}
	\caption{Network structure for selected 18 constituent stock .
	}
	\label{network structure}
\end{figure}

\begin{table}[H]
	\centering
\scalebox{0.7}{
	\begin{tabular}{c|ccc}
		\hline
		\hline
		{ Industry Name} & \multicolumn{1}{c}  {Stock Name } &\\
		\cline{1-3}
		\multirow{1}{*}{Industry}& China State Construction Engineering Co Ltd & CRRC Corporation Limited \\
		\cline{1-3}
		\multirow{2}{*}{Consumer Discretionary}	&  Gree Electric Appliances,Inc. of Zhuhai&Midea Group Co Ltd\\
		\cline{2-3}
		& SAIC Motor Co Ltd \\
		\cline{1-3}
		\multirow{2}{*}{Daily Consumption}& Kweichow Moutai Co Ltd  &Inner Mongolia Yili Industrial Group Co Ltd\\
		\cline{2-3}
		& Wuliangye Yibin Co Ltd \\
		\cline{1-3}
		\multirow{5}{*}{Finance }& Ping An Insurance (Group) Company of China Ltd   & China Merchants Bank Co Ltd\\
		\cline{2-3}
		&Industrial Bank&Bank of Communications Co Ltd\\
		\cline{2-3}
		&China Minsheng Banking Corp Ltd&Agricultural Bank of China Co Ltd\\
		\cline{2-3}
		&CITIC Securities Co Ltd&Shanghai Pudong Development Bank Co Ltd\\
		\cline{2-3}
		&Industrial and Commercial Bank of China Ltd&China Pacific Insurance (Group) Co Ltd\\
		\hline
		\hline
	\end{tabular}
}
	\caption{Stock Name (Divided by the First Classification Criteria of JointQuant) }
	\label{JointQuant1}
\end{table}

We first estimate the model using a rolling window of 88080 observations and then obtain forecasting of $h_{it}, i=1, 2, \cdots, 18$ at future 1-day volatility. Then, we estimate the model directly using the in-sample 88080 high-frequency prices, and forecast the future 1-day volatility.
To illustrate the prediction power of the network HEAVY model, 
we also compute further network GARCH model \citep{Zhou:2020}, referred to as NGARCH model. For out-of-sample model evaluation,  a QLIKE loss function is conducted the forecasting by
\begin{align}
	\mathtt{loss}(r_{i,t+s}^{2}, \tilde{\sigma}_{i,t+s|t}^2)=\frac{r_{i,t+s}^{2}}{\tilde{\sigma}_{i,t+s|t}^2}-\log(\frac{r_{i,t+s}^{2}}{\tilde{\sigma}_{i,t+s|t}^2})-1,
	\label{evaluation model}
\end{align}
where $r_{i,t+s}^{2}$ is the proxy used for the time $t+s$ (latent) variance and $\tilde{\sigma}_{t+s|t}^2$ is a predictor made at time $t$, and  $s=1,\cdots, S$, $i=1, \cdots, N$.  This loss function has been shown to be robust to certain types of noise in the proxy in \cite{Patton:2011a}, \cite{Patton:2011b} and \cite{PattonandSheppard:2009}. For the rolling window procedue, the QLIKE loss function is the average of all forecast values.

The comparison forecasting results of the one-step NHEAVY model,  two-step  NHEAVY model, one-step NGARCH, and two-step NGARCH model with rolling windows procedure and directly in-sample procedure are summarized in Tables \ref{QLIKE_1}  and \ref{QLIKE_2}, respectively.   From Table \ref{QLIKE_1}  and  \ref{QLIKE_2}, one-step NHEAVY model has stronger predicting power than two-step NHEAVY model. Similarly, the one-step NGARCH model has stronger predictor power than the two-step NGARCH model. Sometimes, one-step NGARCH outperforms the two-step NHEAVY model. This is because the two-step NHEAVY model is referred to estimate ${\mu_{R}}_{i}$, $\mu_{i}$, and $\kappa_{i}$, then these estimators are plugged into the quasi-likelihood function, which may change the asymptotic standard error. Therefore, we would compare the one-step NHEAVY and one-step NGARCH model in future two-day and five-day forecasting power.  The results are summarized in Table \ref{QLIKE_3}, 
NHEAVY model has stronger forecasting power than NGARCH model for future 2 days and 5 days.  As a result, it can be seen that NHEAVY model outperforms network GARCH models in short-term forecasting. 

\begin{table}[H]
	\centering
\scalebox{0.62}{
	\begin{tabular}{c|c|c|c|cc}
		\hline
		\hline
		Industry Name & NHEAVY$_{one }$&NHEAVY$_{two }$&NGARCH$_{one }$&NGARCH$_{two}$&\\
		\cline{1-5} 
		Ping An Insurance (Group) Company of China Ltd  & 0.424  &0.698& 0.465&0.835 &\\
		Kweichow Moutai Co Ltd & 0.431 &0.516&0.526&0.722&\\
		China Merchants Bank Co Ltd Moutai Co Ltd & 0.338 &0.377&0.367&0.532&\\
		Industrial Bank& 0.403&0.372&0.429&0.478&\\
		Gree Electric Appliances,Inc. of Zhuhai of Communications Co LTD& 0.436&0.586&0.448&0.662&\\
		Midea Group CO., LTD & 0.330&0.497&0.381&0.562&\\
		Bank of Communications Co LTD& 0.811&0.549&0.971&0.902&\\
		China Minsheng Banking Corp Ltd&0.782&1.116&0.881&1.161&\\
		Inner Mongolia Yili Industrial Group Co Ltd & 0.417&0.473&0.495&0.693&\\
		Agricultural Bank of China Co Ltd& 1.074&0.891&1.174&1.228&\\
		CITIC Securities Co Ltd & 0.444&0.663&0.498&0.758&\\
		China State Construction Engineering Co Ltd & 0.426&1.311&0.473&1.218&\\
		Shanghai Pudong Development Bank Co Ltd& 0.493&0.363&0.516&0.512&\\
		Industrial and Commercial Bank of China Ltd & 0.700&0.632&0.780&0.870&\\
		Wuliangye Yibin Co Ltd& 0.397&0.658&0.442&0.775&\\
		SAIC Motor Co Ltd & 0.458&0.474&0.458&0.608&\\
		China Pacific Insurance (Group) Co Ltd &0.228&0.389&0.282&0.489&\\
		CRRC Corporation Limited & 0.696&0.822&0.744&1.078&\\
		\hline
		\hline
	\end{tabular} 
}
	\caption{QLIKEs for one- step NHEAVY, two-step NHEAVY, one-step NGARCH, two-step NGARCH model using rollling window procedue for one-day forecast.}
	\label{QLIKE_1}
\end{table}

\begin{table}[H]
	\centering
	\scalebox{0.65}{
	\begin{tabular}{c|c|c|c|cc}
		\hline
		\hline
		Industry Name & NHEAVY$_{one }$&NHEAVY$_{two }$&NGARCH$_{one }$&NGARCH$_{two}$&\\
		\cline{1-5} 
		Ping An Insurance (Group) Company of China Ltd  & 0.279  &0.435& 0.522&0.588&\\
		Kweichow Moutai Co Ltd & 0.393&0.394&0.575&0.519&\\
		China Merchants Bank Co Ltd Moutai Co Ltd & 0.959&0.967&1.421&1.394&\\
		Industrial Bank& 1.320&1.198&1.733&1.536&\\
		Gree Electric Appliances,Inc. of Zhuhai of Communications Co LTD& 0.355&0.484&0.405&0.562&\\
		Midea Group CO., LTD & 0.155&0.260&0.220&0.269&\\
		Bank of Communications Co LTD& 0.487&0.140&0.849&0.428&\\
		China Minsheng Banking Corp Ltd&1.779&2.026&2.678&2.574&\\
		Inner Mongolia Yili Industrial Group Co Ltd & 0.620&0.658&0.658&0.683&\\
		Agricultural Bank of China Co Ltd& 1.935&1.758&2.488&2.162&\\
		CITIC Securities Co Ltd & 0.318&0.512&0.332&0.521&\\
		China State Construction Engineering Co Ltd & 0.481&1.231&0.560&1.174&\\
		Shanghai Pudong Development Bank Co Ltd& 1.567&1.302&1.947&1.643&\\
		Industrial and Commercial Bank of China Ltd & 1.450&1.388&2.043&1.869&\\
		Wuliangye Yibin Co Ltd& 0.051&0.079&0.055&0.114&\\
		SAIC Motor Co Ltd & 0.201&0.096&0.441&0.556&\\
		China Pacific Insurance (Group) Co Ltd &0.631&0.834&0.619&0.930&\\
		CRRC Corporation Limited & 0.914&0.973&1.070&1.191&\\
		\hline
		\hline
	\end{tabular} 
}
	\caption{QLIKEs for one- step NHEAVY, two-step NHEAVY, one-step NGARCH, two-step NGARCH model using directly in-sample procedue for one-day forecast.}
	\label{QLIKE_2}
\end{table}

\begin{table}[H]
	\centering
\scalebox{0.6}{
	\begin{tabular}{c|c|c|c|cc}
		\hline
		\hline
		Industry Name & NHEAVY$_{one }$ (2)&NGARCH$_{one }$ (2 )&NHEAVY$_{one }$ (5 )&NGARCH$_{one}$ (5)&\\
		\cline{1-5} 
		Ping An Insurance (Group) Company of China Ltd  & 0.933  &1.588& 0.396&0.980&\\
		Kweichow Moutai Co Ltd & 0.925 &1.439&0.348&0.727&\\
		China Merchants Bank Co Ltd Moutai Co Ltd & 0.543&1.273&2.018&2.828&\\
		Industrial Bank& 0.483&1.130&0.639&1.203&\\
		Gree Electric Appliances,Inc. of Zhuhai of Communications Co LTD& 0.465&0.772&0.276&0.712&\\
		Midea Group CO., LTD & 0.625&0.787&3.194&3.667&\\
		Bank of Communications Co LTD& 0.552&1.178&1.679&2.183&\\
		China Minsheng Banking Corp Ltd&3.260&4.506&1.605&2.519&\\
		Inner Mongolia Yili Industrial Group Co Ltd & 0.060&0.104&2.333&2.684&\\
		Agricultural Bank of China Co Ltd& 1.374&2.238&1.910&2.520&\\
		CITIC Securities Co Ltd & 0.558&0.837&0.162&0.192&\\
		China State Construction Engineering Co Ltd & 0.992&1.378&0.017&0.028&\\
		Shanghai Pudong Development Bank Co Ltd& 1.594&2.332&2.876&3.490&\\
		Industrial and Commercial Bank of China Ltd & 1.027&1.911&1.489&2.196&\\
		Wuliangye Yibin Co Ltd& 0.810&1.112&2.156&2.742&\\
		SAIC Motor Co Ltd & 0.005&0.011&0.169&0.261&\\
		China Pacific Insurance (Group) Co Ltd &0.743&1.044&0.248&0.618&\\
		CRRC Corporation Limited & 0.718&1.160&2.197&2.637&\\
		\hline
		\hline
	\end{tabular} 
}
	\caption{QLIKE for one- step NHEAVY, one-step NGARCH for two-day and five-day forecast values..}
	\label{QLIKE_3}
\end{table}

\section{Conclusion}

This paper  introduces one novel network NHEAVY model, which takes network structure  into consideration. The proposed model reduces the computational complexity substantially from $\mathcal{O}(N^2)$ to $\mathcal{O}(N)$. The parameters in network HEAVY models are estimated by the quasi-likelihood function.  All statistical properties of estimators are presented in Section 3, and the findings are confirmed by numerical studies. We further illustrate  the usefulness of our models using a real data set from the CSI 300 in Chinese stock market. It can show that the network HEAVY model outperforms network GARCH model in short-term forecasting. 

The proposed network HEAVY models can be  extended in further directions. First, it would be interesting to add asymmetric terms to network HEAVY models to explicitly capture the leverage effects and  this may improve forecast performance further. It might also be beneficial to use a long-run/short/run component model in the dynamic equations to separate our transitory movements in volatility.

\section*{Acknowledgements}

Guodong Li's research  is supported in part by HK RGC Grant GRF-17306519.
Junhui Wang's research is supported in part by HK RGC Grants GRF-11303918, GRF-11300919 and GRF-11304520.

\section*{Data Availability Statment}
The data that support the findings of this study are available in JoinQuant Data Services at https://www.joinquant.com/data.

\bibliography{myReferences}

\section* {Appendix}
\textbf{Proof of Theorem 1}
\begin{proof}
We first prove the parameters $\phi=(\omega, \alpha, \lambda, \beta)^{\top}$ and $\phi_R=(\omega_R, \alpha_R, \lambda_R, \beta_R)^{\top}$ satisfy the consistency and asympotics, repectively. Then the statistical properties of the parameters $\theta=( \phi^{\top}, \phi_R^{\top})^{\top}$ are intutive. 
	
We give some notations. Let $K$ and $\rho$ be generic constants taking many different values with $K>0$ and $\rho \in (0,1)$.

By the expressions of $\widehat {h}_{it}(\phi)$, ${h}_{it}(\phi)$, $\widehat{\mu}_{it}(\phi_R)$ and 
${\mu}_{it}(\phi_R)$ in Section 2, we have the following facts
\begin{align*}
\widehat {h}_{it}(\phi)=&\sum_{k=1}^{t-1} \beta^{k-1}\left\{\omega+\alpha {\rm RM}_{i, t-k}+\lambda d_{i}^{-1} \sum_{j \neq i} a_{i j} {\rm RM}_{j, t-k}\right\}+ \beta^{t-1}\widehat{h}_{i1}(\phi), t \geq 2,\\
{h}_{it}(\phi)=&\sum_{k=1}^{t-1} \beta^{k-1}\left\{\omega+\alpha {\rm RM}_{i, t-k}+\lambda d_{i}^{-1} \sum_{j \neq i} a_{i j} {\rm RM}_{j, t-k}\right\}+\beta^{t-1}{h}_{i1}(\phi), t \geq 2,
\end{align*}

\begin{align*}
	\widehat {\mu}_{it}(\phi_R)=&\sum_{k=1}^{t-1} \beta_R^{k-1}\left\{\omega_R+\alpha_R {\rm RM}_{i, t-k}+\lambda_R d_{i}^{-1} \sum_{j \neq i} a_{i j} {\rm RM}_{j, t-k}\right\}+ \beta_R^{t-1}\widehat{\mu}_{i1}(\phi_R), t \geq 2,\\
	{\mu}_{it}(\phi_R)=&\sum_{k=1}^{t-1} \beta_R^{k-1}\left\{\omega_R+\alpha_R {\rm RM}_{i, t-k}+\lambda_R d_{i}^{-1} \sum_{j \neq i} a_{i j} {\rm RM}_{j, t-k}\right\}+\beta_R^{t-1}{\mu}_{i1}(\phi_R), t \geq 2,
\end{align*}
Then, we have
\begin{align}
{h}_{it}(\phi)=&\widehat {h}_{it}(\phi)-\beta^{t-1}\widehat{h}_{i1}(\phi)+\beta^{t-1}{h}_{i1}(\phi), \tag{A1}\\
{\mu}_{it}(\phi_R)=&\widehat {\mu}_{it}(\phi_R)-\beta_R^{t-1}\widehat{\mu}_{i1}(\phi_R)+\beta_R^{t-1}{\mu}_{i1}(\phi_R), \tag{A2}
\end{align}
due to $\widehat {h}_{i1}(\phi)=T^{-1/2}\sum_{t=1}^{(\sqrt{T})}r_{it}^2$,
and $\widehat {\mu}_{i1}(\phi_R)=T^{-1/2}\sum_{t=1}^{(\sqrt{T})}{\rm RM}_{it}$, then
\begin{align}
	\begin{split}
	\frac{\partial {h}_{it}(\phi)}{\partial \phi}=& \frac{\partial \widehat {h}_{it}(\phi)}{\partial \phi}+(t-1) \beta^{t-2} {h}_{i1}(\phi) \mathbf{e}+\beta^{t-1} \frac{\partial {h}_{i1}(\phi)}{\partial \phi},\\
	\frac{\partial^{2}{h}_{it}(\phi)}{\partial \phi \partial \phi^{\top}}=& \frac{\partial^{2} \widehat {h}_{it}(\phi)}{\partial \phi \partial \phi^{\top}}+(t-1)(t-2) \beta^{t-3}{h}_{i1}(\phi) \mathbf{e e}^{\top} \\
	&+2 (t-1) \beta^{t-2} \frac{\partial {h}_{i1}(\phi)}{\partial \phi} \mathbf{e}^{\top}+\beta^{t-1} \frac{\partial^{2} {h}_{i1}(\phi)}{\partial \phi \partial \phi^{\top}},\\
	\vspace{10mm}
	\frac{\partial {\mu}_{it}(\phi_R)}{\partial \phi_R}=& \frac{\partial \widehat {\mu}_{it}(\phi_R)}{\partial \phi_R}+(t-1) \beta^{t-2} {\mu}_{i1}(\phi_R) \mathbf{e}+\beta_R^{t-1} \frac{\partial {\mu}_{i1}(\phi_R)}{\partial \phi_R}, \\
\frac{\partial^{2}{\mu}_{it}(\phi_R)}{\partial \phi_R \partial \phi_R^{\top}}=& \frac{\partial^{2} \widehat {\mu}_{it}(\phi_R)}{\partial \phi_R \partial \phi_R^{\top}}+(t-1)(t-2) \beta^{t-3}{\mu}_{i1}(\phi_R) \mathbf{e e}^{\top} \\
&+2 (t-1) \beta_R^{t-2} \frac{\partial {\mu}_{i1}(\phi_R)}{\partial \phi_R} \mathbf{e}^{\top}+\beta_R^{t-1} \frac{\partial^{2} {\mu}_{i1}(\phi_R)}{\partial \phi_R \partial \phi_R^{\top}},
\end{split}
\tag{A3}
\end{align}
where $\mathbf{e}=(0, 0, 0, 1)^{\top}$, and
\begin{align}
	\begin{split}
	&\sup _{\phi \in \Phi} \frac{1}{{h}_{it}(\phi)} \frac{\partial {h}_{it}(\phi)}{\partial \phi_{k}} \leq K, \\
	&\sup _{\phi \in \Phi} \frac{1}{{h}_{it}(\phi)} \frac{\partial {h}_{it}(\phi)}{\partial \beta} \leq K \sum_{j=1}^{\infty} \rho^{j s}\left\{\sup _{\phi \in \Phi} h_{i, t-j}(\phi)\right\}^{s}, 
	\end{split}
\tag{A4}
\end{align}
where $\phi_{k} \in \{ \omega, \alpha, \lambda\}$, by the inequality $x/ (1+x)\leq x^s$ for $x>0$ and $s\in (0, 1]$.
After the similar calculating, we have
\begin{align}
	\begin{split}
		&\sup _{\phi_R \in \Phi_R} \frac{1}{{\mu}_{it}(\phi_R)} \frac{\partial {\mu}_{it}(\phi_R)}{\partial {\phi_R}_{k}} \leq K, \\
		&\sup _{\phi_R \in \Phi_R} \frac{1}{{\mu}_{it}(\phi_R)} \frac{\partial {\mu}_{it}(\phi_R)}{\partial \beta_R} \leq K \sum_{j=1}^{\infty} \rho^{j s}\left\{\sup _{\phi_R\in \Phi_R} \mu_{i, t-j}(\phi_R)\right\}^{s}, 
	\end{split}
	\tag{A5}
\end{align}
where ${\phi_R}_{k} \in \{ \omega_R, \alpha_R, \lambda_R\}$.
	
We first provide two lemmas for the quasi-likelihood functions $\widehat {L}^{r}$ and ${L}^{r}$.

\begin{lemma}
	 If Assumption \ref{assumption1}-\ref{assumption2} hold, then
	
	(i). $\lim\limits _{T \rightarrow \infty} \sup\limits _{\phi \in \Phi}|L^r(\phi)-\widehat{L}^r(\phi)|=0$ a.s..
	
	(ii). ${\mathrm E}\left|\ell^r_{t}\left(\phi_{0}\right)\right|<\infty$, and ${\mathrm E} \ell_{t}^{r}(\phi) \geq {\mathrm E} \ell_{t}^{r}\left(\phi_{0}\right)$, where the equality holds if and only if $\phi=\phi_{0}$.
	
	(iii). Any $\phi \neq \phi_{0}$ has a neighourhood $V(\phi)$ such that
\begin{align*}
	\liminf _{T \rightarrow \infty} \inf _{\phi^{*} \in V(\phi)} \widehat{L}^r \left(\phi^{*}\right)>\mathrm {E} \ell_{t}^{r}\left(\phi_{0}\right) .
\end{align*}
\end{lemma} 

\begin{proof}
	The proof is similar to that of Theorem 1 in \cite{Zhou:2020}, and thus it is omitted.
\end{proof}

\begin{lemma}
	If the conditions in Theorem 1 hold, then
	
	(i). $\quad \sqrt{N T}\left\|\frac{\partial \widehat{L}^r\left(\phi_{0}\right)}{\partial \phi}-\frac{\partial L^r\left(\phi_{0}\right)}{\partial \phi}\right\|=o(1) \quad a . s .$;
	
	(ii). $\sup\limits_{\left\|\phi-\phi_{0}\right\|<\eta}\left\|\frac{\partial^{2} \widehat{L}^r(\phi)}{\partial \phi \partial \phi^{\top}}-\frac{\partial^{2} L^r\left(\phi_{0}\right)}{\partial \phi \partial \phi^{\top}}\right\|=O_{p}(\eta)$;
	
	(iii). $\frac{\partial^{2} L^r\left(\phi_{0}\right)}{\partial \phi \partial \phi^{\top}} \rightarrow \bm \Sigma^r$, where $\bm\Sigma^r=\frac{1}{N} \sum_{i=1}^{N} {\mathrm E}\left(\frac{1}{h_{i t}^2(\phi_0)} \frac{\partial h_{i t}\left(\phi_{0}\right)}{\partial \phi} \frac{\partial h_{i t}\left(\phi_{0}\right)}{\partial \phi^{\top}}\right)$;
	
	(iv). $\quad \sqrt{N T} \frac{\partial L^r\left(\phi_{0}\right)}{\partial \phi} \stackrel{d}{\longrightarrow} \mathcal{N}\left(0,\left(\kappa_{2}^r-1\right) \bm \Sigma^r\right)$, where $\bm\Sigma^r=\frac{1}{N} \sum_{i=1}^{N} {\mathrm E}\left(\frac{1}{h_{i t}^2(\phi_0)} \frac{\partial h_{i t}\left(\phi_{0}\right)}{\partial \phi} \frac{\partial h_{i t}\left(\phi_{0}\right)}{\partial \phi^{\top}}\right)$, and $\kappa_{2}^{r}={\mathrm E} (\varepsilon_{i t}^2)$.
\end{lemma}	
	\begin{proof}
	By simply calculating the first derivative of $\widehat{L}(\phi)$ and ${L}(\phi)$, we have 
	\begin{align}
		\begin{split}
		\sqrt{N T}\left\|\frac{\partial \widehat{L}^r(\phi)}{\partial \phi}-\frac{\partial L^r(\phi)}{\partial \phi}\right\| & \leq \frac{1}{\sqrt{N T}} \sum_{i=1}^{N} \sum_{t=2}^{T}\left\|\frac{1}{\widehat{h}_{i t}(\phi)} \frac{\partial \widehat{h}_{i t}(\phi)}{\partial \phi}-\frac{1}{h_{i t}(\phi)} \frac{\partial h_{i t}(\phi)}{\partial \phi}\right\| \\
		&+\frac{1}{\sqrt{N T}} \sum_{i=1}^{N} \sum_{t=2}^{T}\left\|\frac{r_{i t}^{2}}{\widehat{h}_{i t}^{2}(\phi)} \frac{\partial \widehat{h}_{i t}(\phi)}{\partial \phi}-\frac{r_{i t}^{2}}{h_{i t}^2(\phi)} \frac{\partial h_{i t}(\phi)}{\partial \phi}\right\|, 
		\end{split}\tag{A6}
	\end{align}
The first term of the right hand in (A6) can have the following results,
\begin{align*}
	\frac{1}{\sqrt{N T}} & \sum_{i=1}^{N} \sum_{t=2}^{T}\left\|\frac{1}{\widehat{h}_{i t}(\phi)} \frac{\partial \widehat{h}_{i t}(\phi)}{\partial \phi}-\frac{1}{h_{i t}(\phi)} \frac{\partial h_{i t}(\phi)}{\partial \phi}\right\| \\
	\leq & \frac{K}{\sqrt{N T}} \sum_{i=1}^{N} \sum_{t=2}^{T}\left\|\frac{\partial \widehat{h}_{i t}(\phi)}{\partial \phi}-\frac{\partial h_{i t}(\phi)}{\partial \phi}\right\| \\
	&+\frac{1}{\sqrt{N T}} \sum_{i=1}^{N} \sum_{t=2}^{T}\left\|\left[h_{i t}(\phi)-\widehat{h}_{i t}(\phi)\right] \frac{1}{h_{i t}(\phi)} \frac{\partial h_{i t}(\phi)}{\partial \phi}\right\|.
\end{align*}
By the facts in (A1), we have
\begin{align*}
	&\frac{1}{\sqrt{N T}} \sum_{i=1}^{N} \sum_{t=2}^{T} \sup _{\theta \in \Theta}\left\|\frac{\partial \widehat{h}_{i t}(\phi)}{\partial \phi}-\frac{\partial h_{i t}(\phi)}{\partial \phi}\right\| \\
	&\leq \frac{1}{\sqrt{N}} \sum_{i=1}^{N} \sup _{\phi \in \Phi} h_{i 1}(\phi)\left(\frac{1}{\sqrt{T}} \sum_{t=2}^{T} (t-1) \beta^{t-2}\right) \\
	&\quad+\frac{1}{\sqrt{N}} \sum_{i=1}^{N} \sup _{\phi \in \Phi}\left\|\frac{\partial h_{i 1}(\phi)}{\partial \phi}\right\|\left(\frac{1}{\sqrt{T}} \sum_{t=2}^{T} \beta^{t-1}\right) \\
	&\rightarrow 0 \quad \text { a.s. }.
\end{align*}
Due to (A4) and Markov's inequality, we can have
\begin{align*}
\sum_{t=2}^{\infty} P\left({\beta^{t-1} \sup _{\phi \in \Phi}}\left\|\frac{1}{h_{i t}(\phi)} \frac{\partial h_{i t}(\phi)}{\partial \phi}\right\|>\epsilon\right) \leq \sum_{t=2}^{\infty} \frac{\beta^{t-1}}{\epsilon} \mathrm {E} \sup _{\phi \in \Phi}\left\|\frac{1}{h_{i t}(\phi)} \frac{\partial h_{i t}(\phi)}{\partial \phi}\right\|<\infty,
\end{align*}
Then, the following results can be easily derived,
\begin{align*}
	&\frac{1}{\sqrt{N T}} \sum_{i=1}^{N} \sum_{t=2}^{T} \sup _{\phi \in \Phi}\left\|\left[h_{i t}(\phi)-\widehat{h}_{i t}(\phi)\right] \frac{1}{h_{i t}(\phi)} \frac{\partial h_{i t}(\phi)}{\partial \phi}\right\| \\
	&\leq \frac{1}{\sqrt{N}} \sum_{i=1}^{N} \sup _{\phi \in \Phi} |h_{i 1}(\phi)- \widehat {h}_{i 1}(\phi)| \left(\frac{1}{\sqrt{T}} \sum_{t=2}^{T} \beta^{t-1} \sup _{\phi \in \Phi}\left\|\frac{1}{h_{i t}(\phi)} \frac{\partial h_{i t}(\phi)}{\partial \phi}\right\|\right) \rightarrow 0 \quad \text { a.s. }
\end{align*}
Thus, the first term of the right hand in (A6) convergence 0 a.s.. Similarly, we can prove the second term of the right hand in (A6) convergence 0 a.s.. Therefore, the proof of (i) is completed.

(ii) We first show
\begin{align*}
\sup _{\left\|\phi-\phi_{0}\right\|<\eta}\left|\frac{\partial^{2} \widehat{L}^r(\phi)}{\partial \beta^{2}}-\frac{\partial^{2} L^r \left(\phi_{0}\right)}{\partial \beta^{2}}\right|=O_{p}(\eta).
\end{align*}
Since
\begin{align*}
	\sup _{\left\|\phi-\phi_{0}\right\|<\eta}\left|\frac{\partial^{2} \widehat{L}^r(\phi)}{\partial \beta^{2}}-\frac{\partial^{2} L^r\left(\phi_{0}\right)}{\partial \beta^{2}}\right| \leq & \frac{1}{T} \sum_{t=2}^{T} \sup _{\phi \in \Phi}\left|\frac{\partial^{2} \widehat{\ell}_{t}^r(\phi)}{\partial \beta^{2}}-\frac{\partial^{2} \ell_{t}^r(\phi)}{\partial \beta^{2}}\right| \\
	&+\frac{1}{T} \sum_{t=2}^{T} \sup _{\left\|\phi-\phi_{0}\right\|<\eta}\left|\frac{\partial^{2} \ell_{t}^r(\phi)}{\partial \beta^{2}}-\frac{\partial^{2} \ell_{t}^r\left(\phi_{0}\right)}{\partial \beta^{2}}\right|.
\end{align*}
It is note that
\begin{align*}
	\frac{\partial^{2} \ell_{t}^r(\phi)}{\partial \beta^{2}}=\frac{1}{N} \sum_{i=1}^{N} &\left\{\frac{1}{h_{i t}(\phi)} \frac{\partial^{2} h_{i t}(\phi)}{\partial \beta^{2}}-\frac{1}{h_{i t}^{2}(\phi)} \frac{\partial h_{i t}(\phi)}{\partial \beta} \frac{\partial h_{i t}(\phi)}{\partial \beta}\right\} \\
	& \quad-\frac{1}{N} \sum_{i=1}^{N}\left\{\frac{r_{i t}^{2}}{h_{i t}^{2}(\phi)} \frac{\partial^{2} h_{i t}(\phi)}{\partial \beta^{2}}-2 \frac{r_{i t}^{2}}{h_{i t}^{3}(\phi)} \frac{\partial h_{i t}(\phi)}{\partial \beta} \frac{\partial h_{i t}(\phi)}{\partial \beta}\right\} .
\end{align*}
Then, we have
\begin{align*}
	\frac{1}{T} \sum_{t=2}^{T} & \sup _{\phi \in \Phi}\left|\frac{\partial^{2} \widehat{\ell}_{t}^r(\phi)}{\partial \beta^{2}}-\frac{\partial^{2} \ell_{t}^r(\phi)}{\partial \beta^{2}}\right| \\
	\leq & \frac{1}{N} \sum_{i=1}^{N} \frac{1}{T} \sum_{t=2}^{T} \sup _{\phi \in \Phi}\left|\frac{1}{h_{i t}(\phi)} \frac{\partial^{2} h_{i t}(\phi)}{\partial \beta^{2}}-\frac{1}{\widehat{h}_{i t}^{2}(\phi)} \frac{\partial^{2} \widehat{h}_{i t}(\phi)}{\partial \beta^{2}}\right| \\
	&+\frac{1}{N} \sum_{i=1}^{N} \frac{1}{T} \sum_{t=2}^{T} \sup _{\phi \in \phi}\left|\frac{1}{h_{i t}^{2}(\phi)} \frac{\partial h_{i t}(\phi)}{\partial \beta} \frac{\partial h_{i t}(\phi)}{\partial \beta}-\frac{1}{\widehat{h}_{i t}^{2}(\phi)} \frac{\partial \widehat{h}_{i t}(\phi)}{\partial \beta} \frac{\partial \widehat{h}_{i t}(\phi)}{\partial \beta}\right| \\
	&+\frac{1}{N} \sum_{i=1}^{N} \frac{1}{T} \sum_{t=2}^{T} \sup _{\theta \in \Theta}\left|\frac{r_{i t}^{2}}{h_{i t}^{2}(\phi)} \frac{\partial^{2} h_{i t}(\phi)}{\partial \beta^{2}}-\frac{r_{i t}^{2}}{\widehat{h}_{i t}^{2}(\phi)} \frac{\partial^{2} \widehat{h}_{i t}(\phi)}{\partial \beta^{2}}\right| \\
	&+\frac{2}{N} \sum_{i=1}^{N} \frac{1}{T} \sum_{t=2}^{T} \sup _{\theta \in \Theta}\left|\frac{r_{i t}^{2}}{h_{i t}^{3}(\phi)} \frac{\partial h_{i t}(\phi)}{\partial \beta} \frac{\partial h_{i t}(\phi)}{\partial \beta}-\frac{r_{i t}^{2}}{\widehat{h}_{i t}^{3}(\phi)} \frac{\partial \widehat{h}_{i t}(\phi)}{\partial \beta} \frac{\partial \widehat{h}_{i t}(\phi)}{\partial \beta}\right| \\
	:=& {\rm I}+\rm {II}+ \rm {III}+\rm {IV}.
\end{align*}
For the term ${\rm I}$, by the facts in (A3), we have
\begin{align}
	\begin{split}
	{\rm I} \leq & \frac{1}{N} \sum_{i=1}^{N} \frac{1}{T} \sum_{t=2}^{T} \sup _{\phi \in \Phi}\left|\frac{\left[h_{i t}(\phi)-\widehat{h}_{i t}(\phi)\right]}{\widehat{h}_{i t}(\phi) h_{i t}(\phi)} \frac{\partial^{2} h_{i t}(\phi)}{\partial \beta^{2}}\right| \\
	&+\frac{K}{N} \sum_{i=1}^{N} \frac{1}{T} \sum_{t=2}^{T} \sup _{\phi \in \Phi}\left|\frac{\partial^{2} h_{i t}(\phi)}{\partial \beta^{2}}-\frac{\partial^{2} \widehat{h}_{i t}(\phi)}{\partial \beta^{2}}\right| \\
	\leq & \frac{1}{N} \sum_{i=1}^{N} \sup _{\phi \in \Phi} |h_{i 1}(\phi) -\widehat {h}_{i 1}(\phi)| \frac{1}{T} \sum_{t=2}^{T} \beta^{t-1} \sup _{\phi \in \Phi}\left|\frac{\partial^{2} h_{i t}(\phi)}{\partial \beta^{2}}\right| \\
	&+\frac{1}{N} \sum_{i=1}^{N} \sup _{\phi \in \Phi} h_{i 1}(\phi)\left(\frac{1}{T} \sum_{t=2}^{T} (t-1)(t-2)\beta^{t-3}\right)\\
	&+\frac{2}{N} \sum_{i=1}^{N} \sup _{\phi \in \Phi}\left|\frac{\partial h_{i 1}(\phi)}{\partial \beta}\right|\left(\frac{1}{T} \sum_{t=1}^{T} (t-1) \beta^{t-2}\right)\\
	&+\frac{1}{N} \sum_{i=1}^{N} \sup _{\phi \in \Phi} \frac{\partial^{2} h_{i 1}(\phi)}{\partial \beta^{2}}\left(\frac{1}{T} \sum_{t=2}^{T} \beta^{t-1}\right) \rightarrow 0\quad \text { as } T \rightarrow \infty.
	\end{split}
\tag{A7}
\end{align}
Similarly, we can prove that $\rm{II}$, ${\rm III}$ and ${\rm IV}$ converge to 0 a.s.. Thus, 
\begin{align}
\frac{1}{T} \sum_{t=2}^{T} \sup _{\phi \in \Phi}\left|\frac{\partial^{2} \widehat{\ell}_{t}^r(\phi)}{\partial \beta^{2}}-\frac{\partial^{2} \ell_{t}^r(\phi)}{\partial \beta^{2}}\right| \rightarrow 0 \quad \text { as } T \rightarrow \infty. \tag{A8}
\end{align}
On the other hand, by the Taylor expansion, we have
\begin{align*}
\mathrm {E }\sup _{\left\|\phi-\phi_{0}\right\|<\eta}\left|\frac{\partial^{2} \ell_{t}^r(\phi)}{\partial \beta^{2}}-\frac{\partial^{2} \ell_{t}^r\left(\phi_{0}\right)}{\partial \beta^{2}}\right| \leq \eta \mathrm {E} \sup _{\left\|\phi-\phi_{0}\right\|<\eta}\left|\frac{\partial^{3} \ell_{t}^r(\phi)}{\partial \beta^{3}}\right|,
\end{align*}
By a simple calculation, it follows that
\begin{align*}
	\frac{\partial^{3} \ell_{t}(\phi)}{\partial \beta^{3}}=& \frac{1}{N} \sum_{i=1}^{N}\left\{2-\frac{6 r_{i t}^{2}}{h_{i t}(\phi)}\right\}\left[\frac{1}{h_{i t}(\phi)} \frac{\partial h_{i t}(\phi)}{\partial \beta}\right]^{3} \\
	&+\frac{1}{N} \sum_{i=1}^{N}\left\{1-\frac{r_{i t}^{2}}{h_{i t}(\phi)}\right\} \frac{1}{h_{i t}(\phi)} \frac{\partial^{3} h_{i t}(\phi)}{\partial \beta^{3}} \\
	&+\frac{1}{N} \sum_{i=1}^{N}\left\{\frac{6 r_{i t}^{2}}{h_{i t}(\phi)}-3\right\} \frac{1}{h_{i t}(\phi)} \frac{\partial h_{i t}(\phi)}{\partial \beta} \frac{1}{h_{i t}(\phi)} \frac{\partial^{2} h_{i t}(\phi)}{\partial \beta^{2}} .
\end{align*}
Similar to the proof of Lemma 2 in \cite{Zhou:2020}, it is not hard to show that 
$\mathrm {E} \sup\limits_{\left\|\phi-\phi_{0}\right\|<\eta}\left|\frac{\partial^{3} \ell_{t}(\phi)}{\partial \beta^{3}}\right|=O(1)$, which together with (A8), we can prove 
\begin{align*}
\sup _{\left\|\phi-\phi_{0}\right\|<\eta}\left|\frac{\partial^{2} \widehat{L}(\phi)}{\partial \beta^{2}}-\frac{\partial^{2} L\left(\phi_{0}\right)}{\partial \beta^{2}}\right|=O_{p}(\eta).
\end{align*}
Similary, we can show that
\begin{align*}
\sup _{\left\|\phi-\phi_{0}\right\|<\eta}\left|\frac{\partial^{2} \widehat{L}(\phi)}{\partial \phi_{i} \partial \phi_{j}}-\frac{\partial^{2} L\left(\phi_{0}\right)}{\partial \phi_{i} \partial \phi_{j}}\right|=O_{p}(\eta)
\end{align*}
where $\phi_{i}, \phi_{j} \in\{\omega, \alpha, \lambda, \beta\}$. Thus, (ii) holds.

(iii). By the simple calculation, we have
\begin{align*}
\frac{\partial^{2} \ell_{t}^r\left(\phi_{0}\right)}{\partial \phi \partial \phi^{\top}}=&\frac{1}{N} \sum_{i=1}^{N}\left\{\left(1-\varepsilon_{it}\right) \frac{1}{h_{i t}(\phi_0)} \frac{\partial^{2} h_{i t}\left(\phi_{0}\right)}{\partial \phi_{0} \partial \phi^{\top}}\right.\\
&\left.+\left(2 \varepsilon_{it}-1\right) \frac{1}{h_{i t}^{2}(\phi_0)} \frac{\partial h_{i t}\left(\phi_{0}\right)}{\partial \phi} \frac{\partial h_{i t}\left(\phi_{0}\right)}{\partial \phi^{\top}}\right\}.
\end{align*}
Since ${\varepsilon}_{it}$ is i.i.d with $\mathrm{E} \{{ \varepsilon}_{it}|\mathcal {F}_{t-1}^{\rm HF}\}=1$  and $\mathrm{var} \{{ \varepsilon}_{it}|\mathcal {F}_{t-1}^{\rm HF}\}={{\kappa}}_2^{r}$, by the strong law of large numbers, (iii) holds.

(iv). Note that
\begin{align*}
\sqrt{N T} \frac{\partial L^r\left(\phi_{0}\right)}{\partial \phi}=\frac{1}{\sqrt{N T}} \sum_{t=2}^{T} \sum_{i=1}^{N}\left(1-\varepsilon_{i t}\right) \frac{1}{h_{i t}(\phi_0) }\frac{\partial h_{i t}\left(\phi_{0}\right)}{\partial \phi}.
\end{align*}
By the martingale central limit theorem, we have
\begin{align*}
\sqrt{N T} \frac{\partial L^{r}\left(\phi_{0}\right)}{\partial \phi} \stackrel{d}{\longrightarrow} \mathcal{N}\left(0,\left(\kappa_{2}^{r}-1\right) \bm\Sigma^{r}\right),
\end{align*}
where $\bm\Sigma^r=\frac{1}{N} \sum_{i=1}^{N} {\mathrm E}\left(\frac{1}{h_{i t}^2(\phi_0)} \frac{\partial h_{i t}\left(\phi_{0}\right)}{\partial \phi} \frac{\partial h_{i t}\left(\phi_{0}\right)}{\partial \phi^{\top}}\right)$, and $\kappa_{2}^{r}={\mathrm E} (\varepsilon_{i t}^2)$.

\begin{lemma}
	If Assumption \ref{assumption1}-\ref{assumption2} hold, then
	
	(i). $\lim\limits _{T \rightarrow \infty} \sup\limits _{\phi_R \in \Phi_R}|L^{\rm RM}(\phi_R)-\widehat{L}^{\rm RM}(\phi_R)|=0$ a.s..
	
	(ii). ${\mathrm E}\left|\ell^{\rm RM}_{t}\left({\phi_{R}}_{0}\right)\right|<\infty$, and ${\mathrm E} \ell_{t}^{\rm RM}(\phi_R) \geq {\mathrm E} \ell_{t}^{\rm RM}\left({\phi_{R}}_{0}\right)$, where the equality holds if and only if $\phi_R={\phi_R}_{0}$.
	
	(iii). Any $\phi_{R} \neq {\phi_{R}}_{0}$ has a neighourhood $V(\phi_R)$ such that
	\begin{align*}
		\liminf _{T \rightarrow \infty} \inf _{\phi_{R}^{*} \in V(\phi_{R})} \widehat{L}^{\rm RM} \left(\phi_{R}^{*}\right)>\mathrm {E} \ell_{t}^{\rm RM}\left({\phi_{R}}_{0}\right) .
	\end{align*}
\end{lemma} 

\begin{proof}
	The proof is similar to that of Lemma 1.
\end{proof}

\begin{lemma}
	If the conditions in Theorem 1 hold, then
	
	(i). $\quad \sqrt{N T}\left\|\frac{\partial \widehat{L}\left({\phi_R}_{0}\right)}{\partial \phi}-\frac{\partial L\left({\phi_R}_{0}\right)}{\partial \phi}\right\|=o(1) \quad a . s .$;
	
	(ii). $\sup\limits_{\left\|\phi_R-{\phi_R}_{0}\right\|<\eta}\left\|\frac{\partial^{2} \widehat{L}(\phi_R)}{\partial \phi \partial \phi_{R}^{\top}}-\frac{\partial^{2} L\left({\phi_{R}}_{0}\right)}{\partial \phi_R \partial \phi_R^{\top}}\right\|=O_{p}(\eta)$;
	
	(iii). $\frac{\partial^{2} L\left({\phi_R}_{0}\right)}{\partial \phi_R \partial \phi_R^{\top}} \rightarrow \bm \Sigma^{R}$, where $\bm\Sigma^R=\frac{1}{N} \sum_{i=1}^{N} {\mathrm E}\left(\frac{1}{u_{i t}^2({\phi_R}_0)} \frac{\partial u_{i t}\left({\phi_R}_{0}\right)}{\partial \phi_R} \frac{\partial u_{i t}\left({\phi_R}_{0}\right)}{\partial \phi_{R}^{\top}}\right)$;
	
	(iv). $\quad \sqrt{N T} \frac{\partial L\left({\phi_{R}}_{0}\right)}{\partial \phi_R} \stackrel{d}{\longrightarrow} \mathcal{N}\left(0,\left(\kappa_{2}^{R}-1\right) \bm \Sigma^{R}\right)$ where $\kappa_{2}^{R}={\mathrm E} (\epsilon_{i t}^2)$, 
	
	\quad \quad and $\bm\Sigma^R=\frac{1}{N} \sum_{i=1}^{N} {\mathrm E}\left(\frac{1}{u_{i t}^2({\phi_R}_0)} \frac{\partial u_{i t}\left({\phi_R}_{0}\right)}{\partial \phi_R} \frac{\partial u_{i t}\left({\phi_R}_{0}\right)}{\partial \phi_{R}^{\top}}\right)$.
\end{lemma}	

\begin{proof}
The proof is similar in Lemma 2.
\end{proof}
\textbf{Proof of Theorem 1}
By a standard compactness argument, using Lemma 1 and Lemma 2, we have
\begin{align*}
	0=\frac{\partial {\widehat L}^{r}(\widehat{\phi})}{\partial \phi}=\frac{\partial {\widehat L}^{r}(\phi_0)}{\partial \phi}+\frac{\partial^2 {\widehat L}^{r}(\phi^{*})}{\partial \phi\partial \phi^{\top}}(\widehat{\phi}-\phi_0)
\end{align*}
where $\phi^{*}$ lies in $\widehat{\phi}$ and $\phi_0$, thus
\begin{align*}
	\sqrt{NT}(\widehat{\phi}-\phi_0)=&-\left(\frac{\partial^2 \widehat {L}^{r}(\phi^{*})}{\partial \phi\partial \phi^{\top}}  \right)^{-1}\sqrt{NT}\frac{\partial \widehat L^{r}(\phi_0)}{\partial \phi}\\
	=&-\left(\frac{\partial^2 L^{r}(\phi^{*})}{\partial \phi\partial \phi^{\top}}  \right)^{-1}\sqrt{NT}\frac{\partial L^{r}(\phi_0)}{\partial \phi}+o_p(1)\\
\stackrel{d}{\longrightarrow}& \mathcal{N}\left(0,\left(\kappa_{2}^r-1\right) {\bm\Sigma_{r}}^{-1}\right)
\end{align*}
Similarly, using Lemma 3 and Lemma 4, we have
\begin{align*}
	0=\frac{\partial {\widehat L}^{\rm RM}(\widehat{\phi}_R)}{\partial \phi_R}=\frac{\partial {\widehat L}^{r}({\phi_R}_0)}{\partial \phi_R}+\frac{\partial^2 {\widehat L}^{\rm RM}(\phi_R^{*})}{\partial \phi_R \partial \phi_R^{\top}}(\widehat{\phi}_R-{\phi_R}_{0})
\end{align*}
where $\phi_R^{*}$ lies in $\widehat{\phi}_R$ and ${phi_R}_0$, thus
\begin{align*}
	\sqrt{NT}(\widehat{\phi}_R-{\phi_R}_0)=&-\left(\frac{\partial^2 \widehat {L}^{\rm RM}(\phi_R^{*})}{\partial \phi_R\partial \phi_R^{\top}}  \right)^{-1}\sqrt{NT}\frac{\partial \widehat L^{\rm RM}({\phi_R}_0)}{\partial \phi_R}\\
	=&-\left(\frac{\partial^2 L^{\rm RM}(\phi_R^{*})}{\partial \phi_R \partial \phi_R^{\top}}  \right)^{-1}\sqrt{NT}\frac{\partial L^{\rm RM}({\phi_R}_0)}{\partial \phi_R}+o_p(1)\\
\stackrel{d}{\longrightarrow}& \mathcal{N}\left(0,\left(\kappa_{2}^{R}-1\right) {\bm\Sigma_{R}}^{-1}\right)
\end{align*}
Since $\theta=( \phi^{\top}, \phi_R^{\top})^{\top}$, by an straightforword application of quasi-likelihood theory \citep{Bollers:1992,Cipollini:2007}, we have
	\begin{align*}
	\sqrt{NT}(\widehat{\theta}-\theta_{0}) \stackrel{d}{\longrightarrow}N(0,\mathcal {I}^{-1} \mathcal {J}{( \mathcal {I}^{-1})}^{\top}),
\end{align*}
where
\begin{align*}
	\mathcal {I}= \begin{pmatrix}
		&{\bm\Sigma_{r}}& \mathbf{0}\\
		&\mathbf{0} & {\bm\Sigma_{R}}
	\end{pmatrix},
	\quad
	\mathcal{J}=\begin{pmatrix}
		&\left(\kappa_{2}^r-1\right) {\bm\Sigma_{r}}& \left(\kappa_{2}^{r,R}-1\right) {\bm\Sigma_{r,R}}\\
		& \left(\kappa_{2}^{r,R}-1\right) {\bm\Sigma_{r,R}}&(\kappa_{2}^{R}-1) {\bm\Sigma_{R}}
	\end{pmatrix}.
\end{align*}

\end{proof}
\end{proof}

\end{document}